\documentclass{lmcs} 
\pdfoutput=1

\usepackage[utf8]{inputenc}

\usepackage{lastpage}
\lmcsdoi{17}{3}{22}
\lmcsheading{}{\pageref{LastPage}}{}{}%
{Dec.~01,~2020}{Sep.~02,~2021}{}


\theoremstyle{plain} 

\usepackage{amssymb}
\usepackage{cite}
\usepackage{calc}
\usepackage{tikz}
\usepackage{multicol}
\usetikzlibrary{positioning,chains,fit,shapes,calc}
\usepackage{subcaption}

\newcommand{\revone}[1]{{\color{black}#1}}
\newcommand{\revtwo}[1]{{\color{black}#1}}
\newcommand{\revthree}[1]{{\color{black}#1}}

\usepackage{algorithm,algorithmicx}
\usepackage[noend]{algpseudocode}
\usepackage{xspace}
\usepackage{color, colortbl}

\usepackage{ifmtarg}
\makeatletter
\newcommand{\toprule}{\hrule height.8pt depth0pt \kern2pt} 
\newcommand{\midrule}{\kern2pt\hrule\kern2pt} 
\newcommand{\bottomrule}{\kern2pt\hrule\relax}
\newcommand{\algcaption}[2][]{%
  \refstepcounter{algorithm}%
  \@ifmtarg{#1}
    {\addcontentsline{loa}{figure}{\protect\numberline{\thealgorithm}{\ignorespaces #2}}}
    {\addcontentsline{loa}{figure}{\protect\numberline{\thealgorithm}{\ignorespaces #1}}}%
  \toprule
  \textbf{\fname@algorithm~\thealgorithm}\ #2\par 
  \midrule
}
\makeatother

\def\e#1{\emph{#1}}
\def\scs{\mathbf{S}}
\def\ar{\mathit{ar}}
\def\const{\mathsf{Const}}
\def\pall{\mathord{\mathcal{P}}}

\def\ra{\rightarrow}

\def\reals{\mathbb{R}}

\newcommand{\db}{\mathrm{DB}}

\def\exo{_{\mathsf{x}}}
\def\endo{_{\mathsf{n}}}

\def\shapley{\mathrm{Shapley}}

\newcommand*{\MyDef}{\mathrm{def}}
\newcommand*{\eqdefU}{\ensuremath{\mathop{\overset{\MyDef}{=}}}}
\newcommand*{\eqdef}{\mathrel{\overset{\MyDef}{\resizebox{\widthof{\eqdefU}}{\heightof{=}}{=}}}}

\def\set#1{\mathord{\{#1\}}}

\def\tup#1{\vec{#1}}
\def\dl{\mathrel{{:}{\text{-}}}}

\def\qcnt{\mathsf{count}}
\def\qsum{\mathsf{sum}}

\def\qmin{\mathsf{min}}
\def\qmax{\mathsf{max}}

\def\is{\mathsf{IS}}
\def\ds{\mathsf{NIS}}

\def\cqrst{\mathsf{q_{\mathsf{RST}}}}

\def\rst{\mathsf{RST}}
\def\var{\mathsf{Vars}}
\def\at{\mathsf{Atoms}}

\def\Pr#1{\mathord{\mathrm{Pr}\left[#1\right]}}

\def\val#1{\mathtt{#1}}
\def\rel#1{\textsc{#1}}
\def\att#1{\textit{#1}}

\newcommand{\algname}[1]{{\sf #1}}
\def\myrulewidth{3.20in}

\def\therule{\rule{\myrulewidth}{0.2pt}}

\newenvironment{insidecode}[3]
{
\begin{tabular}{p{\myrulewidth}}
\multicolumn{1}{c}{\rule{0mm}{3mm}{\bf #3} $\algname{#1}(\mbox{#2})$\vspace{-0.6em}}\\
\therule\vskip-0.8em\therule
\vspace{0em}
\begin{algorithmic}[1]}
{\end{algorithmic}
\vskip-0.3em\therule
\end{tabular}}

\newenvironment{subroutine}
{\begin{algorithm}\floatname{algorithm}{Subroutine}}
{\end{algorithm}}

\newcounter{subroutine}

\newenvironment{malgorithm}[2]
{
\algcaption{#1}\label{#2}
\begin{algorithmic}[1]
}
{
\end{algorithmic}
\bottomrule
}

\newtheorem{lemma}[thm]{Lemma}

\def\dblcurlyl{\{\!\!\{}
\def\dblcurlyr{\}\!\!\}}

\def\angs#1{\mathord{\langle #1 \rangle}}
\def\bracs#1{\mathord{[#1]}}

\newcommand{\eat}[1]{}

\def\auf{f^{\mathrm{a}}}
\def\instf{f^{\mathrm{i}}}
\def\pubf{f^{\mathrm{p}}}
\def\citf{f^{\mathrm{c}}}

\newenvironment{repeatresult}[2]
{\vskip0.5em\par\textbf{#1 #2.}\em}
{\vskip1em}

\def\SAT{CntSat\xspace}

\def\fpsharpp{\mathrm{FP}^{\mathrm{\#P}}}
\def\sharpp{\mathrm{\#P}}
\def\result{\mathrm{result}}

\begin{document}

\title{The Shapley Value of Tuples in Query Answering}

\author[E.~Livshits]{Ester Livshits\rsuper{a}}	

\author[L.~Bertossi]{Leopoldo Bertossi\rsuper{b}}	

\author[B.~Kimelfeld]{Benny Kimelfeld\rsuper{a}}	

\author[M.~Sebag]{Moshe Sebag\rsuper{a}}	

\address{\lsuper{a}Technion, Haifa, Israel}	
\email{esterliv@cs.technion.ac.il, bennyk@cs.technion.ac.il, moshesebag@campus.technion.ac.il}  
\address{\lsuper{b}Univ.~Adolfo Iba\~nez and Millenium Inst. Foundations of Data (IMFD), Chile }	
\email{leopoldo.bertossi@uai.cl}  




\begin{abstract}
  \noindent  We investigate the application of the Shapley value to quantifying
  the contribution of a tuple to a query answer. The Shapley value is
  a widely known numerical measure in cooperative game theory and in
  many applications of game theory for assessing the contribution of a
  player to a coalition game. It has been established already in the
  1950s, and is theoretically justified by being the very single
  wealth-distribution measure that satisfies some natural
  axioms. While this value has been investigated in several areas, it
  received little attention in data management. We study this measure
  in the context of conjunctive and aggregate queries by defining
  corresponding coalition games. We provide algorithmic and
  complexity-theoretic results on the computation of Shapley-based
  contributions to query answers; and for the hard cases we present
  approximation algorithms.
\end{abstract}

\maketitle


\section{Introduction}
The Shapley value is named after Lloyd Shapley who introduced the
value in a seminal 1952 article~\cite{shapley:book1952}. He considered
a \e{cooperative game} that is played by a set $A$ of players and is
defined by a \e{wealth function} $v$ that assigns, to each coalition
$S\subseteq A$, the wealth $v(S)$. For instance, in our running
example the players are researchers, and $v(S)$ is the total number of
citations of papers with an author in $S$. As another example, $A$
might be a set of politicians, and $v(S)$ the number of votes that a
poll assigns to the party that consists of the candidates in $S$. The
question is how to distribute the wealth $v(A)$ among the players, or
from a different perspective, how to quantify the contribution of each
player to the overall wealth. For example, the removal of a researcher
$r$ may have zero impact on the overall number of citations, since
each paper has co-authors from $A$. Does it mean that $r$ has no
contribution at all?  What if the removal in turns of \e{every}
individual author has no impact?  Shapley considered distribution
functions that satisfy a few axioms of good behavior. Intuitively,
the axioms state that the function should be invariant under
isomorphism, the sum over all players should be equal to the total
wealth, and the contribution to a sum of wealths is equal to the sum
of separate contributions.  Quite remarkably, Shapley has established
that there is a \e{single} such function, and this function has become
known as the \e{Shapley value}.

The Shapley value is informally defined as follows. Assume that we
select players one by one, randomly and without replacement, starting
with the empty set. Whenever we select the player $p$, its addition to
the set $S$ of players selected so far may cause a change in wealth
from $v(S)$ to $v(S\cup\set{p})$. The Shapley value of $p$ is the
expectation of change that $p$ causes in this probabilistic process.

The Shapley value has been applied in various areas and fields beyond
cooperative game theory
(e.g.,~\cite{aumann2003endogenous,DBLP:conf/stacs/AzizK14}), such as
bargaining foundations in economics~\cite{gul1989bargaining}, takeover
corporate rights in law~\cite{nenova2003value}, pollution
responsibility in environmental
management~\cite{petrosjan2003time,liao2015case}, influence
measurement in social network analysis~\cite{narayanam2011shapley},
and utilization of multiple Internet Service Providers (ISPs) in
networks~\cite{ma2010internet}. Closest to database manegement is the
application of the Shapley value to attributing a level of
\e{inconsistency} to a statement in an inconsistent knowledge
base~\cite{DBLP:journals/ai/HunterK10,DBLP:conf/ijcai/YunVCB18}; the
idea is natural: as wealth, adopt a measure of inconsistency for a set
of logical sentences~\cite{DBLP:journals/jiis/GrantH06}, and then
associate to each sentence its Shapley value. {\color{black}The Shapley value has also become relevant in machine learning, as the SHAP approach to providing numerical scores to feature values of entities under classification~\cite{NIPS2017_8a20a862,Lundberg2020}.}

In this article, we apply the Shapley value to quantifying the
contribution of database facts (tuples) to query results. As in
previous work on quantification of contribution of
facts~\cite{DBLP:journals/pvldb/MeliouGMS11,DBLP:conf/tapp/SalimiBSB16},
we view the database as consisting of two types of facts:
\e{endogenous} facts and \e{exogenous} facts. Exogenous facts are
taken as given (e.g., inherited from external sources) without
questioning, and are beyond experimentation with hypothetical or
counterfactual scenarios. On the other hand, we may have control over
the endogenous facts, and these are the facts for which we reason
about existence and marginal contribution.
Our focus is on queries that can be viewed as mapping databases to
numbers. These include Boolean queries (mapping databases to zero and
one) and aggregate queries (e.g., count the number of tuples in a
multiway join). As a cooperative game, the endogenous facts take the
role of the players, and the result of the query is the wealth. The
core computational problem for a query is then: given a database and
an endogenous fact, compute the Shapley value of the fact.

We study the complexity of computing the Shapley value for Conjunctive
Queries (CQs) and aggregate functions over CQs.  Our main results are
as follows.  We first establish a dichotomy in data complexity for the
class of Boolean CQs without self-joins. Interestingly, our dichotomy
is the same as that of query inference in tuple-independent probabilistic databases~\cite{DBLP:conf/vldb/DalviS04}: if the CQ
is hierarchical, then the problem is solvable in polynomial time, and
otherwise, it is $\fpsharpp$-complete (i.e., complete for the
intractable class of polynomial-time algorithms with an oracle to,
e.g., a counter of the satisfying assignments of a propositional
formula).  The proof, however, is more challenging than that of Dalvi
and Suciu~\cite{DBLP:conf/vldb/DalviS04}, as the Shapley value
involves coefficients that do not seem to easily factor out.  Since
the Shapley value is a probabilistic expectation, we show how to use
the linearity of expectation to extend the dichotomy to arbitrary summations
over CQs without self-joins. For non-hierarchical queries (and, in
fact, all unions of CQs), we show that both Boolean and summation
versions are efficiently approximable (i.e., have a multiplicative
FPRAS) via Monte Carlo sampling.

The general conclusion is that computing the exact Shapley value is
notoriously hard, but the picture is optimistic if approximation is
allowed under strong guarantees of error boundedness. Our results
immediately generalize to non-Boolean CQs and group-by operators,
where the goal is to compute the Shapley value of a fact to each
tuple in the answer of a query. For aggregate functions other than
summation (where we cannot apply the linearity of expectation), the
picture is far less complete, and remains for future
investigation. Nevertheless, we give some positive preliminary results
about special cases of the minimum and maximum aggregate functions.

Various formal measures have been proposed for quantifying the
contribution of a fact $f$ to a query answer.  Meliou et
al.~\cite{DBLP:journals/pvldb/MeliouGMS11} adopted the quantity of
\e{responsibility} that is inversely proportional to the minimal
number of endogenous facts that should be removed to make $f$
counterfactual (i.e., removing $f$ transitions the answer from true to
false). This measure adopts earlier notions of formal causality by
Halpern and Pearl~\cite{DBLP:conf/uai/HalpernP01}. This measure,
however, is fundamentally designed for non-numerical queries, and it
is not at all clear whether it can incorporate the numerical
contribution of a fact (e.g., recognizing that some facts contribute
more than others due to high numerical attributes).  Salimi et
al.~\cite{DBLP:conf/tapp/SalimiBSB16} proposed the \e{causal effect}:
assuming endogenous facts are randomly removed independently and
uniformly, what is the difference in the expected query answer between
assuming the presence and the absence of $f$? Interestingly, as we
show here, this value is the same as the \e{Banzhaf power index} that
has also been studied in the context of wealth distribution in
cooperative games~\cite{10.2307/3689345}, and is different from the
Shapley value~\cite[Chapter 5]{roth1988shapley}. {\color{black} We also show that,
with the exception of a multiplicative approximation, 
our complexity results
extend to  the 
causal-effect measure. While the
justification to measuring fact contribution using one measure over
the other is yet to be established, we believe that the suitability of
the Shapley value is backed by the aforementioned theoretical
justification as well as its massive adoption in a plethora of
fields.

This article is the full version of a conference publication~\cite{DBLP:conf/icdt/LivshitsBKS20}.
We have added all of the
proofs and intermediate results that were excluded from the 
paper. In particular, we have added the full proof of our main result---the dichotomy in the complexity of computing the Shapley value for Boolean CQs (Theorem~\ref{thm:bcq-dichotomy}), as well as the proofs of our results for aggregate queries over CQs (Theorem~\ref{thm:agg-hard} and Proposition~\ref{prop:max}). We also provide a complexity analysis for the computation of the causal-effect measure in Section~\ref{sec:banzhaf} and, consequently, confirm a conjecture we made in the conference publication, stating that our complexity results are also applicable to the causal-effect measure (and Banzhaf power index).
}

The remainder of the article is organized as follows. In the next
section, we give preliminary concepts, definitions and notation.  In
Section~\ref{sec:shapley}, we present the Shapley value to measure the
contribution of a fact to a query answer, along with illustrating
examples. In Section~\ref{sec:complexity}, we study the complexity of
calculating the Shapley value. Finally, we discuss past contribution
measures in Section~\ref{sec:measures} and conclude in
Section~\ref{sec:conclusions}.

\section{Preliminaries}

\label{sec:preliminaries}
\paragraph*{Databases} A (relational) \e{schema} $\scs$ is a
collection of \e{relation symbols} with each relation symbol $R$ in
$\scs$ having an associated arity that we denote by $\ar(R)$.  We
assume a countably infinite set $\const$ of \e{constants} that are
used as database values. If $\tup c=(c_1,\dots,c_k)$ is a tuple of
constants and $i\in\set{1,\dots,k}$, then we use $\tup c[i]$ to refer
to the constant $c_i$. A \e{relation} $r$ is a set of tuples of
constants, each having the same arity (length) that we denote by
$\ar(r)$.  A \e{database} $D$ (over the schema $\scs$) associates with
each relation symbol $R$ a finite relation
$R^D$, such that $\ar(R)=\ar(R^D)$. We denote by $\db(\scs)$ the set
of all databases over the schema $\scs$.  Notationally, we identify a
database $D$ with its finite set of \e{facts} $R(c_1,\dots,c_k)$,
stating that the relation $R^D$ over the $k$-ary relation symbol $R$
contains the tuple $(c_1,\dots,c_k)\in\const^k$.  In particular, two
databases $D$ and $D'$ over $\scs$ satisfy $D\subseteq D'$ if and only
if $R^D\subseteq R^{D'}$ for all relation symbols $R$ of $\scs$.

Following previous work on the explanations and responsibility of facts to
query
answers~\cite{DBLP:conf/mud/MeliouGMS10,DBLP:journals/debu/MeliouGHKMS10},
we view the database as consisting of two types of facts:
\e{exogenous} facts and \e{endogenous} facts. Exogenous facts
represent a context of information that is taken for granted and
assumed not to claim any contribution or responsibility to the result
of a query. Our concern is about \e{the role of the endogenous facts}
in establishing the result of the query. In notation, we denote by
$D\exo$ and $D\endo$ the subsets of $D$ that consist of the exogenous
and endogenous facts, respectively.  Hence, in our notation we have
that $D=D\exo\cup D\endo$.

{
\definecolor{Gray}{gray}{0.9}
\def\emprow{\multicolumn{3}{l}{}}
\begin{figure}[t]
\centering
\begin{subfigure}[b]{0.24\linewidth}
\begin{tabular}{r|c|c|} 
\cline{1-3}
\rowcolor{Gray}
\multicolumn{3}{l}{\rel{Author} (endo)}\\\cline{2-3}
 & $\att{name}$ & $\att{affil}$ \\\cline{2-3}
$\auf_1$ & $\val{Alice}$ & $\val{UCLA}$\\
$\auf_2$ & $\val{Bob}$ & $\val{NYU}$\\
$\auf_3$ & $\val{Cathy}$ & $\val{UCSD}$\\
$\auf_4$ & $\val{David}$ & $\val{MIT}$\\
$\auf_5$ & $\val{Ellen}$ & $\val{UCSD}$\\
\cline{2-3}
\emprow
\end{tabular}
\end{subfigure}
\begin{subfigure}[b]{0.23\linewidth}
\begin{tabular}{c|c|c|} 
\cline{1-3}
\rowcolor{Gray}
\multicolumn{3}{l}{\rel{Inst} (exo)}\\\cline{2-3}
& $\att{name}$ & $\att{state}$ \\\cline{2-3}
$\instf_1$ & $\val{UCLA}$ & $\val{CA}$\\
$\instf_2$ & $\val{UCSD}$ & $\val{CA}$\\
$\instf_3$ & $\val{NYU}$ & $\val{NY}$\\
$\instf_4$ & $\val{MIT}$ & $\val{MA}$\\
\cline{2-3}
\emprow\\
\emprow
\end{tabular}
\end{subfigure}
\begin{subfigure}[b]{0.24\linewidth}
\begin{tabular}{c|c|c|} 
\cline{1-3}
\rowcolor{Gray}
\multicolumn{3}{l}{\rel{Pub} (exo)}\\\cline{2-3}
& $\att{author}$ & $\att{pub}$ \\\cline{2-3}
$\pubf_1$ & $\val{Alice}$ & $\val{A}$\\
$\pubf_2$ & $\val{Alice}$& $\val{B}$\\
$\pubf_3$ & $\val{Bob}$ & $\val{C}$\\
$\pubf_4$ & $\val{Cathy}$ & $\val{C}$\\
$\pubf_5$ & $\val{Cathy}$ & $\val{D}$\\
$\pubf_6$ & $\val{David}$ & $\val{C}$\\
\cline{2-3}
\end{tabular}
\end{subfigure}
\,
\begin{subfigure}[b]{0.24\linewidth}
\begin{tabular}{c|c|c|} 
\cline{1-3}
\rowcolor{Gray}
\multicolumn{3}{l}{\rel{Citations} (exo)}\\\cline{2-3}
& $\att{paper}$ & $\att{cits}$ \\\cline{2-3}
$\citf_1$ & $\val{A}$ & $\val{18}$\\
$\citf_2$ & $\val{B}$ & $\val{2}$\\
$\citf_2$ & $\val{C}$ & $\val{8}$\\
$\citf_3$ & $\val{D}$ & $\val{12}$\\
\cline{2-3}
\emprow\\
\emprow
\end{tabular}
\end{subfigure}
\caption{\label{fig:DB} The database of the running example. }
\end{figure}
}

\begin{exa}\label{example:db}
  Figure~\ref{fig:DB} depicts the database $D$ of our running example
  from the domain of academic publications. The relation $\rel{Author}$
  stores authors along with their affiliations, which are stored with
  their states in $\rel{Inst}$. The relation $\rel{Pub}$ associates
  authors with their publications, and $\rel{Citations}$ stores the
  number of citations for each paper\footnote{\revone{This example is used for illustrative purposes only and does not express any suggestion of a way to rank researchers.}}. For example, publication
  $\val{C}$ has $\val{8}$ citations and it is written jointly by
  $\val{Bob}$ from $\val{NYU}$ of $\val{NY}$ state,
  $\val{Cathy}$ from $\val{UCSD}$ of $\val{CA}$ state, and $\val{David}$ from $\val{MIT}$ of $\val{MA}$ state. All
  $\rel{Author}$ facts are endogenous, and all remaining facts are
  exogenous. Hence, $D\endo=\set{\auf_1,\auf_2,\auf_3,\auf_4,\auf_5}$ and $D\exo$
  consists of all $f^x_j$ for
  $x\in\set{\mathrm{i},\mathrm{p},\mathrm{c}}$ and relevant $j$.\qed
\end{exa}

\paragraph*{Relational and conjunctive queries}
Let $\scs$ be a schema.  A \e{relational query} is a function that
maps databases to relations. More formally, a relational query $q$ of
arity $k$ is a function $q: \db(\scs)\ra \pall(\const^{k})$ (where $\pall(\const^k)$ is the power set of $\const^k$ that consists
of all subsets of $\const^k$) that maps every
database over $\scs$ to a finite relation $q(D)$ of arity $k$. We
denote the arity of $q$ by $\ar(q)$.  Each tuple $\tup c$ in $q(D)$ is
an \e{answer} to $q$ on $D$.  If the arity of $q$ is zero, then we say
that $q$ is a \e{Boolean} query; in this case, $D\models q$ denotes
that $q(D)$ consists of the empty tuple $()$, while $D\not\models q$
denotes that $q(D)$ is empty.

Our analysis will focus on the special case of \e{Conjunctive Queries}
(CQs). A CQ over the schema $\scs$ is a relational query definable
by a first-order formula of the form $\exists {y}_1\cdots \exists
{y}_m \theta(\tup{x}, {y}_1, \ldots, {y}_m)$, where
$\theta$ is a conjunction of atomic formulas of the form $R(\tup t)$ with variables among
those in $\tup{x}, {y}_1, \ldots, {y}_m$.  In the remainder of
the article, a CQ $q$ will be written shortly as a logic rule, that is, an
expression of the form
$$
q(\tup x) \dl R_1(\tup{t}_1),\dots,R_n(\tup{t}_n)
$$
where each $R_i$ is a relation symbol of $\scs$, each $\tup{t}_i$ is a
tuple of variables and constants with the same arity as $R_i$, and
$\tup{x}$ is a tuple of $k$ variables from
$\tup{t}_1,\dots,\tup{t}_n$.  We call $q(\tup x)$ the \e{head} of $q$,
and $R_1(\tup{t}_1),\dots,R_n(\tup{t}_n)$ the \e{body} of $q$. Each
$R_i(\tup t_i)$ is an \e{atom} of $q$. The variables occurring in the
head are called the \e{head variables}, and we make the standard
safety assumption that every head variable occurs at least once in the
body. The variables occurring in the body but not in the head are
existentially quantified, and are called the \e{existential
  variables}.  The answers to $q$ on a database $D$ are the tuples
$\tup c$ that are obtained by projecting all homomorphisms
from $q$ to $D$ onto the variables of $\tup x$, and replacing each variable with the constant it is
mapped to. A homomorphism from $q$ to $D$ is a mapping of the variables in $q$ to the constants of $D$, such that every atom in $q$ is mapped to a fact in $D$.

A \e{self-join} in a CQ $q$ is a pair of distinct atoms over the same
relation symbol. For example, in the query
$q() \dl R(x,y),S(x),R(y,z)$, the first and third atoms constitute a
self-join.  We say that $q$ is \e{self-join-free} if it has no self-joins, or in other words, every relation symbol occurs at most once in
the body.

Let $q$ be a CQ. For a variable $y$ of $q$, let $A_y$ be the set of
atoms $R_i(\tup{t}_i)$ of $q$ that contain $y$ (that is, $y$ occurs in
$\tup{t}_i$). We say that $q$ is \e{hierarchical} if for all
existential variables $y$ and $y'$ it holds that
$A_y\subseteq A_{y'}$, or $A_{y'}\subseteq A_y$, or
$A_y\cap A_{y'}=\emptyset$~\cite{DBLP:journals/cacm/DalviRS09}.  
For example, every CQ with at most two atoms is hierarchical. The
smallest non-hierarchical CQ is the following.
\begin{equation}\label{eq:cqrst}
\cqrst() \dl R(x), S(x,y), T(y)
\end{equation}
On the other hand, the query $q(x)\dl R(x), S(x,y), T(y)$, which has a
single existential variable $y$, is hierarchical.

Let $q$ be a Boolean query and $D$ a
database, both over the same schema, and let $f\in D\endo$ be an
endogenous fact. We say that $f$ is a \e{counterfactual cause} (\e{for
  $q$
  w.r.t.~$D$})~\cite{DBLP:journals/debu/MeliouGHKMS10,DBLP:journals/pvldb/MeliouGMS11}
if the removal of $f$ causes $q$ to become false; that is, $D\models
q$ and $D\setminus\set{f}\not\models q$.

\begin{exa}\label{example:cq}
We will use the following queries in our examples.
\begin{align*}
  q_1()\dl & \,\, \rel{Author}(x,y),\rel{Pub}(x,z) \\
  q_2()\dl & \,\, \rel{Author}(x,y),\rel{Pub}(x,z),\rel{Citations}(z,w)\\
  q_3(z,w)\dl & \,\, \rel{Author}(x,y),\rel{Pub}(x,z),\rel{Citations}(z,w)\\
  q_4(z,w)\dl & \,\, \rel{Author}(x,y),\rel{Pub}(x,z),\rel{Citations}(z,w), \rel{Inst}(y,\val{CA})
\end{align*}
Note that $q_1$ and $q_2$ are Boolean, whereas $q_3$ and $q_4$ are
not. Also note that $q_1$ and $q_3$ are hierarchical, and $q_2$ and
$q_4$ are not.  Considering the database $D$ of Figure~\ref{fig:DB},
none of the $\rel{Author}$ facts is a counterfactual cause for $q_1$,
since the query remains true even if the fact is removed. The same
applies to $q_2$. However, the fact $\auf_1$ is a counterfactual cause
for the Boolean CQ $q_1'()\dl
\rel{Author}(x,\val{UCLA}),\rel{Pub}(x,z)$, asking whether there is a
publication with an author from UCLA, since $D$ satisfies $q_1'$, but
if we remove Alice from the database, the query $q_1'$ is not longer satisfied, as no other
author from UCLA exists.  \qed
\end{exa}

\paragraph*{Numerical and aggregate-relational queries}
A \e{numerical query} $\alpha$ is a function that maps databases to
numbers.  More formally, a numerical query $\alpha$ is a function
$\alpha:\db(\scs)\ra \reals$ that maps every database $D$ over $\scs$
to a real number $\alpha(D)$. 
%

A special form of a numerical query $\alpha$ is what we refer to as an
\e{aggregate-relational query}: a $k$-ary relational query $q$
followed by an aggregate function $\gamma:\pall(\const^k)\ra\reals$
that maps the resulting relation $q(D)$ into a single number
$\gamma(q(D))$. We denote this aggregate-relational query as
$\gamma\bracs{q}$; hence, $\gamma\bracs{q}(D)\eqdef\gamma(q(D))$.

Special cases of aggregate-relational queries include the functions of
the form $\gamma=F\angs{\varphi}$ that transform every tuple $\tup c$ into
a number $\varphi(\tup c)$ via a \e{feature function}
$\varphi:\const^k\ra\reals$, and then contract the resulting bag of
numbers into a single number (hence, $F$ is a numerical function on bags of numbers). Formally, we define
$F\angs{\varphi}\bracs{q}(D)\eqdef F(\dblcurlyl\varphi(\tup c)\mid \tup c\in
q(D)\dblcurlyr)$ where $\dblcurlyl\cdot\dblcurlyr$ is used for bag
notation. For example, if we assume that the $i$th attribute of
$q(D)$ takes a numerical value, then $\varphi$ can simply copy this number (i.e.,
$\varphi(\tup c)=\tup c[i]$); we denote this $\varphi$ by $[i]$. As another
example, $\varphi$ can be the product of two attributes:
$\varphi=[i]\cdot[j]$. We later refer to the following
aggregate-relational queries.
\begin{align*}
\qsum\angs{\varphi}\bracs{q}(D)& \eqdef \sum_{\tup c\in q(D)}\varphi(\tup c)\\
\qmax\angs{\varphi}\bracs{q}(D)& \eqdef 
\begin{cases}
\max\set{\varphi(\tup c)\mid \tup c\in q(D)} & \mbox{if $q(D)\neq\emptyset$;}\\
0 & \mbox{if $q(D)=\emptyset$.}\\
\end{cases}
\end{align*}
Other popular examples include the minimum (defined analogously to maximum), average and median over
the feature values. A special case of $\qsum\angs{\varphi}\bracs{q}$ is
$\qcnt\bracs{q}$ that counts the number of answers for $q$. That is,
$\qcnt\bracs{q}$ is $\qsum\angs{\mathbf{1}}\bracs{q}$, where ``$\mathbf{1}$''
is the feature function that maps every $k$-tuple to the
number $1$.  A special case of $\qcnt\bracs{q}$ is when $q$ is Boolean;
in this case, we may abuse the notation and identify $\qcnt\bracs{q}$
with $q$ itself. Put differently, we view $q$ as the numerical query
$\alpha$ defined by $\alpha(D)=1$ if $D\models q$ and $\alpha(D)=0$ if
$D\not\models q$.

\begin{exa}\label{example:agg}
  Following are examples of aggregate-relational queries over the
  relational queries of Example~\ref{example:cq}.
\begin{itemize}
\item  $\alpha_1 \eqdef \qsum\angs{[2]}\bracs{q_3}$ calculates the
  total number of citations of all published papers  with an author in the database.
\item  $\alpha_2 \eqdef \qcnt\bracs{q_3}$ counts the
  papers in $\rel{Citations}$ with an author in the database.
\item $\alpha_3 \eqdef \qsum\angs{[2]}\bracs{q_4}$ calculates the total
  number of citations of papers by Californians.
\item $\alpha_4 \eqdef \qmax\angs{[2]}\bracs{q_3}$ calculates the
  number of citations for the most cited paper.
\end{itemize}
For $D$ of Figure~\ref{fig:DB} we have 
$\alpha_1(D)=\val{40}$, 
$\alpha_2(D)=\val{4}$,
$\alpha_3(D)=\val{40}$ and
$\alpha_4(D)=\val{18}$.\qed
\end{exa}

In terms of presentation, when we mention general functions $\gamma$
and $\varphi$, we make the implicit assumption that they are computable in
polynomial time with respect to the representation of their input.
Also, observe that our modeling of an aggregate-relational query does
not allow for \e{grouping}, since a database is mapped to a single
number.  This is done for simplicity of presentation, and all concepts
and results of this article generalize to grouping as in traditional
modeling (e.g.,~\cite{DBLP:journals/jacm/CohenNS07}). This is
explained in the next section.

\paragraph*{Shapley value} Let $A$ be a finite set of \e{players}.
A \e{cooperative game} is a function $v:\pall(A)\ra \reals$, such that
$v(\emptyset)=0$.  The value $v(S)$ represents a value, such as
wealth, jointly obtained by $S$ when the players of $S$ cooperate. The
\e{Shapley value}~\cite{shapley:book1952} measures the share of each
individual player $a\in A$ in the gain of $A$ for the cooperative game
$v$. Intuitively, the gain of $a$ is as follows. Suppose that we form
a team by taking the players one by one, randomly and uniformly
without replacement; while doing so, we record the change of $v$ due to the addition of $a$ as the random
contribution of $a$. Then
the Shapley value of $a$ is the expectation of the random
contribution.
\begin{equation}\label{eq:shapley-generic}
\shapley(A,v,a)\eqdef \frac{1}{|A|!}\sum_{\sigma\in \Pi_A}\big(v(\sigma_a\cup\set{a})-v(\sigma_a)\big)
\end{equation}
where $\Pi_A$ is the set of all possible permutations over the players in $A$, and for each permutation $\sigma$ we denote by $\sigma_a$ the set of players that appear before $a$ in the permutation.

An alternative formula for the Shapley value is the following.
\begin{equation}\label{eq:shapley-AB}
   \shapley(A,v,a)\eqdef \sum_{B\subseteq
  A\setminus\set{a}}\hskip-0.5em\frac{|B|!\cdot (|A|-|B|-1)!}{|A|!}
\Big(v(B\cup\set{a})-v(B)\Big)
\end{equation}
Note that $|B|!\cdot (|A|-|B|-1)!$ is the number of permutations over
$A$ such that all players in $B$ come first, then $a$, and then all
remaining players. For further reading, we refer the reader to the
book by Roth~\cite{roth1988shapley}.

\section{Shapley Value of Database Facts}\label{sec:shapley}
Let $\alpha$ be a numerical query over a schema $\scs$, and let $D$ be
a database over $\scs$. We wish to quantify the contribution of every
endogenous fact to the result $\alpha(D)$. For that, we view $\alpha$
as a cooperative game over $D\endo$, where the value of every subset
$E$ of $D\endo$ is $\alpha(E\cup D\exo)$.

\begin{defi}[Shapley Value of Facts]\label{def:shapley}
  Let $\scs$ be a schema, $\alpha$ a numerical query, $D$ a database,
  and $f$ an endogenous fact of $D$. The \e{Shapley value} of $f$ for
  $\alpha$, denoted $\shapley(D,\alpha,f)$, is the value
  $\shapley(A,v,a)$ as given in~\eqref{eq:shapley-generic}, where:
\begin{itemize}
\item $A=D\endo$;
\item $v(E)=\alpha(E\cup D\exo)-\alpha(D\exo)$ for all $E\subseteq A$;
\item $a=f$.
\end{itemize}
That is, $\shapley(D,\alpha,f)$ is the Shapley value of $f$ in the
cooperative game that has the endogenous facts as the set of players
and values each team by the quantity it adds to $\alpha$.
\end{defi}

The choice of $v$ is natural in that the first term collects answers where endogenous facts may interact with exogenous facts, but we remove those answers that come only from exogenous facts. As a special case, if $q$ is a Boolean query, then $\shapley(D,q,f)$
is the same as the value $\shapley(D,\qcnt[q],f)$. In this case, the
corresponding cooperative game takes the values $0$ and $1$, and the
Shapley value then coincides with the \e{Shapley-Shubik
  index}~\cite{RePEc:cup:apsrev:v:48:y:1954:i:03:p:787-792_00}.  Some
fundamental properties of the Shapley value~\cite{shapley:book1952} are reflected here as
follows:
\begin{itemize}
\item $\shapley(D,a\cdot\alpha+b\cdot\beta,f)=a\cdot\shapley(D,\alpha,f)+b\cdot\shapley(D,\beta,f)$.
\item $\alpha(D)=\alpha(D\exo)+\sum_{f\in D\endo}\shapley(D,\alpha,f)$.
\end{itemize}

\begin{rem}
  Note that $\shapley(D,\alpha,f)$ is defined for a general numerical
  query $\alpha$. The definition is immediately extendible to queries
  with \e{grouping} (producing tuples of database constants and
  numbers~\cite{DBLP:journals/jacm/CohenNS07}), where we would measure
  the responsibility of $f$ for an answer tuple $\tup a$ and write
  something like $\shapley(D,\alpha,\tup a,f)$. In that case, we treat
  every group as a separate numerical query.  We believe that focusing
  on numerical queries (without grouping) allows us to keep the
  presentation considerably simpler while, at the same time, retaining
  the fundamental challenges.  \qed
\end{rem}

In the remainder of this section, we illustrate the Shapley value on 
our running example.

\begin{exa}\label{ex:bcq-ptime}
  We begin with a Boolean CQ, and specifically $q_1$ from
  Example~\ref{example:cq}. Recall that the endogenous facts
  correspond to the authors.  As Ellen has no publications, her
  addition to any $D\exo\cup E$ where $E\subseteq D\endo$ does not
  change the satisfaction of $q_1$. Hence, its Shapley value is zero:
  $\shapley(D,q_1,\auf_5)=0$. The fact $\auf_1$ changes the query
  result if it is either the first fact in the permutation, or it is the second fact after $\auf_5$. There are $4!$ permutations that satisfy
  the first condition, and $3!$ permutations that satisfy the second.
  The contribution of $\auf_1$ to the query result is one in each of
  these permutations, and zero otherwise. Therefore, we have
  $\shapley(D,q_1,\auf_1)=\frac{4!+3!}{120}=\frac14$. The same argument
  applies to $\auf_2$, $\auf_3$ and $\auf_4$, and so,
  $\shapley(D,q_1,\auf_2)=\shapley(D,q_1,\auf_3)=\shapley(D,q_1,\auf_4)=\frac14$.  We get the
  same numbers for $q_2$, since every paper is mentioned in the
  $\rel{Citations}$ relation. Note that the value of the query $q_1$ on the database is $1$, and it holds that $\sum_{i=1}^5 \shapley(D,q_1,\auf_i)=4\cdot\frac14+0=1$; hence, the second fundamental property of the Shapley value mentioned above is satisfied.
   
  While Alice, Bob, Cathy and David have the same Shapley value for $q_1$,
  things change if we consider the relation $\rel{pub}$ endogenous as
  well: the Shapley value of Alice and Cathy will be higher than Bob's and David's
  values, since they have more publications. Specifically, the fact
  $\auf_1$, for example, will change the query result if and only if
  at least one of $\pubf_1$ or $\pubf_2$ appears earlier in the
  permutation, and no pair among $\set{\auf_2,\pubf_3}$,
  $\set{\auf_3,\pubf_4}$, $\set{\auf_3,\pubf_5}$, and $\set{\auf_4,\pubf_6}$ appears earlier
  than $\auf_1$. By rigorous counting, we can show that there are: $2$
  such sets of size one, $17$ such sets of size two, $56$ such sets of size three, $90$ such sets of size four, $73$ such sets of size five, $28$ such sets of size six, and $4$ such sets of size seven.  Therefore,
  the Shapley value of $\auf_1$ is:
  \begin{align*}
    \shapley(D,q_1,\auf_1)&=2\cdot\frac{(11-2)!1!}{11!}+17\cdot\frac{(11-3)!2!}{11!}+56\cdot\frac{(11-4)!3!}{11!}+90\cdot\frac{(11-5)!4!}{11!}\\
    &+73\cdot\frac{(11-6)!5!}{11!}+28\cdot\frac{(11-7)!6!}{11!}+4\cdot\frac{(11-8)!7!}{11!}=\frac{442}{2520}
  \end{align*}
  We can similarly compute the Shapley value for the rest of the
  authors, concluding that $\shapley(D,q_1,\auf_2)=\shapley(D,q_1,\auf_4)=\frac{241}{2520}$
  and $\shapley(D,q_1,\auf_3)=\frac{442}{2520}$. Hence, the Shapley
  value is the same for Alice and Cathy, who have two publications
  each, and lower for Bob and David, that have only one publication.  
\qed
\end{exa}

The following example, taken from Salimi et
al.~\cite{DBLP:conf/tapp/SalimiBSB16}, illustrates the Shapley value
on (Boolean) graph reachability.
\begin{exa} \label{ex:datalogq} Consider the following database
  $G$ defined via the relation symbol $\rel{Edge}/2$.
\begin{center}
\vskip-0.5em
\scalebox{1.2}{\begin{picture}(0,0)%
\includegraphics{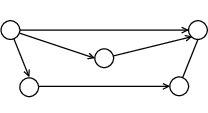}%
\end{picture}%
\setlength{\unitlength}{3947sp}%
\begingroup\makeatletter\ifx\SetFigFont\undefined%
\gdef\SetFigFont#1#2#3#4#5{%
  \reset@font\fontsize{#1}{#2pt}%
  \fontfamily{#3}\fontseries{#4}\fontshape{#5}%
  \selectfont}%
\fi\endgroup%
\begin{picture}(1666,886)(68,-425)
\put(1276, 14){\makebox(0,0)[lb]{\smash{{\SetFigFont{9}{10.8}{\familydefault}{\mddefault}{\updefault}{\color[rgb]{0,0,0}$e_3$}%
}}}}
\put(901,314){\makebox(0,0)[b]{\smash{{\SetFigFont{9}{10.8}{\familydefault}{\mddefault}{\updefault}{\color[rgb]{0,0,0}$e_1$}%
}}}}
\put(901,-361){\makebox(0,0)[b]{\smash{{\SetFigFont{9}{10.8}{\familydefault}{\mddefault}{\updefault}{\color[rgb]{0,0,0}$e_5$}%
}}}}
\put(226,-61){\makebox(0,0)[rb]{\smash{{\SetFigFont{9}{10.8}{\familydefault}{\mddefault}{\updefault}{\color[rgb]{0,0,0}$e_4$}%
}}}}
\put(1651,-61){\makebox(0,0)[lb]{\smash{{\SetFigFont{9}{10.8}{\familydefault}{\mddefault}{\updefault}{\color[rgb]{0,0,0}$e_6$}%
}}}}
\put(601, 14){\makebox(0,0)[rb]{\smash{{\SetFigFont{9}{10.8}{\familydefault}{\mddefault}{\updefault}{\color[rgb]{0,0,0}$e_2$}%
}}}}
\put(901,-24){\makebox(0,0)[b]{\smash{{\SetFigFont{9}{10.8}{\familydefault}{\mddefault}{\updefault}{\color[rgb]{0,0,0}$c$}%
}}}}
\put(301,-256){\makebox(0,0)[b]{\smash{{\SetFigFont{9}{10.8}{\familydefault}{\mddefault}{\updefault}{\color[rgb]{0,0,0}$d$}%
}}}}
\put(1501,-249){\makebox(0,0)[b]{\smash{{\SetFigFont{9}{10.8}{\familydefault}{\mddefault}{\updefault}{\color[rgb]{0,0,0}$e$}%
}}}}
\put(1651,201){\makebox(0,0)[b]{\smash{{\SetFigFont{9}{10.8}{\familydefault}{\mddefault}{\updefault}{\color[rgb]{0,0,0}$b$}%
}}}}
\put(151,201){\makebox(0,0)[b]{\smash{{\SetFigFont{9}{10.8}{\familydefault}{\mddefault}{\updefault}{\color[rgb]{0,0,0}$a$}%
}}}}
\end{picture}%
}
\vskip-0.5em
\end{center}
Here, we assume that all edges $e_i$ are endogenous facts.  Let
$p_{ab}$ be the Boolean query (definable in, e.g., Datalog) that
determines whether there is a path from $a$ to $b$.  Let us calculate
$\shapley(G,p_{ab},e_i)$ for different edges $e_i$.  Intuitively, we
expect $e_1$ to have the highest value since it provides a direct path
from $a$ to $b$, while $e_2$ contributes to a path only in the
presence of $e_3$, and $e_4$ enables a path only in the presence of
both $e_5$ and $e_6$. We show that, indeed, it holds that
$\shapley(G,p_{ab},e_1) > \shapley(G,p_{ab},e_2) >
\shapley(G,p_{ab},e_4)$.

{\color{black}To illustrate the calculation, observe that there are $2^5$ subsets of
$G$ that do not contain $e_1$, and among them, the subsets that
satisfy $p_{ab}$ are the supersets of $\set{e_2,e_3}$ and
$\set{e_4,e_5,e_6}$. Hence, we have that:  
\begin{eqnarray*}
\shapley(G,p_{ab},e_1) &=& \frac{(6 - 0 - 1)! \times 0!}{6!} + 5 \times\frac{(6 - 1 - 1)!\times 1!}{6!} +\\ && \left({{5}\choose{2}}-1\right)\times \frac{(6-2-1)!2!}{6!} + \left({{5}\choose{3}}-4\right)\times \frac{(6-3-1)!3!}{6!}\\
&=& \frac{1}{6} + 5\times\frac{1}{30} + (10 -1)\times \frac{1}{60} + (10 -4) \times \frac{1}{60} = \frac{35}{60}
\end{eqnarray*}

Similarly, there are $2^5$ subsets of
$G$ that do not contain $e_2$, and among them, the subsets that
satisfy $p_{ab}$ are the supersets of $\set{e_1}$ and
$\set{e_4,e_5,e_6}$. Then, 
\begin{eqnarray*}
\shapley(G,p_{ab},e_2) &=& \frac{(6-1-1)! 1!}{6!} + 3\times \frac{(6-2-1)! 2!}{6!} + 3\times \frac{(6-3-1)! 3!}{6!} = \frac{8}{60}
\end{eqnarray*}
A similar reasoning shows that  
$\shapley(G,p_{ab},e_3)=\frac{8}{60}$. Finally, 
among the $2^5$ subsets of $G$ that do not contain $e_4$, those that satisfy $p_{ab}$ are the supersets of $\set{e_1}$ and
$\set{e_2,e_3}$, and it holds that:
\begin{eqnarray*}
\shapley(G,p_{ab},e_4) &=&  \frac{(6-2-1)! 2!}{6!} + 2\times \frac{(6-3-1)! 3!}{6!} = \frac{3}{60}
\end{eqnarray*}
Similarly, $\shapley(G,p_{ab},e_5)=\shapley(G,p_{ab},e_6)=\frac{3}{60}$.}
\qed\end{exa}

Lastly, we consider aggregate functions over conjunctive queries.

\begin{exa}\label{example:count}
  We consider the queries $\alpha_1$, $\alpha_2$, and $\alpha_4$ from
  Example~\ref{example:agg}. Ellen has no publications; hence,
  $\shapley(D,\alpha_j,\auf_5)=0$ for $j\in\set{1,2,4}$. The
  contribution of $\auf_1$ is the same in every permutation ($20$ for
  $\alpha_1$ and $2$ for $\alpha_2$) since Alice is the single author
  of two published papers that have a total of $20$ citations. Hence,
  $\shapley(D,\alpha_1,\auf_1)=20$ and
  $\shapley(D,\alpha_2,\auf_1)=2$. The total number of citations of
  Cathy's papers is also $20$; however, Bob and David are her
  coauthors on paper C. Hence, if the fact $\auf_3$ appears before
  $\auf_2$ and $\auf_4$ in a permutation, its contribution to the query
  result is $20$ for $\alpha_1$ and $2$ for $\alpha_2$, while if
  $\auf_3$ appears after at least one of $\auf_2$ or $\auf_4$ in a
  permutation, its contribution is $12$ for $\alpha_1$ and $1$ for
  $\alpha_2$. Clearly, $\auf_2$ appears before both $\auf_3$ and
  $\auf_4$ in one-third of the permutations. Thus, we have that
  $\shapley(D,\alpha_1,\auf_3)=\frac{1}{3}\cdot 20+\frac{2}{3}\cdot
  12=\frac{44}{3}$ and $\shapley(D,\alpha_2,\auf_3)=\frac{1}{3}\cdot
  2+\frac{2}{3}\cdot 1=\frac{4}{3}$. Using similar computations we
  obtain that
  $\shapley(D,\alpha_1,\auf_2)=\shapley(D,\alpha_1,\auf_4)=\frac83$
  and
  $\shapley(D,\alpha_2,\auf_2)=\shapley(D,\alpha_2,\auf_4)=\frac13$.

We conclude that the Shapley value of Alice, who is the single author of two papers with a total of $20$ citations, is higher than the Shapley value of Cathy who also has two papers with a total of $20$ citations, but shares one paper with other authors. Bob and David have the same Shapley value, since they share a single paper, and this value is the lowest among the four, as they have the lowest number of papers and citations.

Finally, consider $\alpha_4$. The contribution of $\auf_1$ in this
case depends on the maximum value before adding $\auf_1$ in the
permutation (which can be $0$, $8$ or $12$). For example, if $\auf_1$
is the first fact in the permutation, its contribution is $18$ since
$\alpha_4(\emptyset)=0$. If $\auf_1$ appears after $\auf_3$, then its
contribution is $6$, since $\alpha_4(S)=12$ whenever $\auf_3\in S$. We
have that $\shapley(D,\alpha_4,\auf_1)=10$,
$\shapley(D,\alpha_4,\auf_2)=\shapley(D,\alpha_4,\auf_4)=2$ and
$\shapley(D,\alpha_4,\auf_3)=4$ (we omit the computations here). We
see that the Shapley value of $\auf_1$ is much higher than the rest,
since Alice significantly increases the maximum value when added to
any prefix. If the number of citations of paper $\val{C}$ increases to
$16$, then $\shapley(D,\alpha_4,\auf_1)=6$, hence lower. This is
because the next highest value is closer; hence, the contribution of
$\auf_1$ diminishes.
\qed
\end{exa}

\section{Complexity Results}\label{sec:complexity}
In this section, we give complexity results on the computation of the
Shapley value of facts. We begin with exact evaluation for Boolean CQs
(Section~\ref{sec:bcq}), then move on to exact evaluation on
aggregate-relational queries (Section~\ref{sec:agg}), and finally
discuss approximate evaluation (Section~\ref{sec:approx}).  In the
first two parts we restrict the discussion to CQs without self-joins,
and leave the problems open in the presence of self-joins. However,
the approximate treatment in the third part covers the general class
of CQs (and beyond).

\subsection{Boolean Conjunctive Queries}\label{sec:bcq}
We investigate the problem of computing the (exact)
Shapley value w.r.t.~a Boolean CQ without self-joins. Our main result
in this section is a full classification of (i.e., a dichotomy in) the
data complexity of the problem. As we show, the classification
criterion is the same as that of query evaluation over
tuple-independent probabilistic
databases~\cite{DBLP:conf/vldb/DalviS04}: hierarchical CQs without
self-joins are tractable, and non-hierarchical ones are intractable.

\begin{thm}\label{thm:bcq-dichotomy}
  Let $q$ be a Boolean CQ without self-joins. If $q$ is hierarchical,
  then computing $\shapley(D,q,f)$ can be done in polynomial time, given $D$
  and $f$ as input. Otherwise, the problem is $\fpsharpp$-complete.
\end{thm}

Recall that $\fpsharpp$ is the class of functions computable in
polynomial time with an oracle to a problem in $\sharpp$ (e.g.,
counting the number of satisfying assignments of a propositional
formula).  This complexity class is considered intractable, and is
known to be above the polynomial hierarchy (Toda's
theorem~\cite{DBLP:journals/siamcomp/Toda91}).

\begin{exa}
  Consider the query $q_1$ from Example~\ref{example:cq}. This query
  is hierarchical; hence, by Theorem~\ref{thm:bcq-dichotomy},
  $\shapley(D,q_1,f)$ can be calculated in polynomial time, given $D$
  and $f$. On the other hand, the query $q_2$ is not hierarchical. Thus, Theorem~\ref{thm:bcq-dichotomy} asserts that computing
  $\shapley(D,q_2,f)$ is $\fpsharpp$-complete.\qed
\end{exa}

In the rest of this subsection, we discuss the proof of
Theorem~\ref{thm:bcq-dichotomy}. While the tractability condition is
the same as that of Dalvi and Suciu~\cite{DBLP:conf/vldb/DalviS04}, it
is not clear whether and/or how we can use their dichotomy to
prove ours, in each of the two directions (tractability and
hardness). The difference is mainly in that they deal with a random
subset of probabilistically independent (endogenous) facts, whereas we
reason about random \e{permutations} over the facts.  We start by discussing the algorithm for computing the Shapley value in
the hierarchical case, and then we discuss the
proof of hardness for the non-hierarchical case.

\def\CF{\mathrm{CF}}
\def\Sat{\mathrm{Sat}}

\subsubsection{Tractability side} Let $D$ be a database, let $f$ be an
endogenous fact, and let $q$ be a Boolean query. The computation of
$\shapley(D,q,f)$ easily reduces to the problem of counting the
$k$-sets (i.e., sets of size $k$) of endogenous facts that, along with the exogenous facts,
satisfy $q$. More formally, the reduction is to the problem of
computing $|\Sat(D,q,k)|$ where $\Sat(D,q,k)$ is the set of all
subsets $E$ of $D\endo$ such that $|E|=k$ and $(D\exo\cup E)\models q$.
The reduction is based on the following formula, where we denote $m=|D\endo|$ and slightly
abuse the notation by viewing $q$ as a 0/1-numerical query, where
$q(D')=1$ if and only if $D'\models q$.

\begin{align}
  &\shapley(D,q,f) = \sum_{E\subseteq
    (D\endo\setminus\set{f})}\hskip-0.5em\frac{|E|!  (m-|E|-1)!}{m!}
  \Big(q(D\exo\cup E\cup\set{f})-q(D\exo\cup E)\Big) \label{eq:shapley-bcq}\\
  &=\sum_{E\subseteq
      (D\endo\setminus\set{f})}\hskip-1.5em\frac{|E|!  (m-|E|-1)!}{m!}
    \Big(q(D\exo\cup E\cup\set{f})\Big)- \hskip-0.5em \sum_{E\subseteq
      (D\endo\setminus\set{f})}\hskip-1.5em\frac{|E|!  (m-|E|-1)!}{m!}
    \Big(q(D\exo\cup E)\Big)\notag\\
  & =
  \left(\sum_{k=0}^{m-1}\frac{k!(m-k-1)!}{m!}\times|\Sat(D',q,k)|\right)
    -\left(\sum_{k=0}^{m-1}\frac{k!(m-k-1)!}{m!}\times|\Sat(D\setminus\set{f},q,k)|\right)
    \notag
\end{align}
In the last expression, $D'$ is the same as $D$, except that $f$ is
viewed as \e{exogenous} instead of \e{endogenous}.  Hence, to prove
the positive side of Theorem~\ref{thm:bcq-dichotomy}, it suffices to
show the following.

\def\ptimesat{
Let $q$ be a hierarchical Boolean CQ without self-joins. There is a
  polynomial-time algorithm for computing the number $|\Sat(D,q,k)|$
  of subsets $E$ of $D\endo$ such that $|E|=k$ and $(D\exo\cup E)\models q$,
  given $D$ and $k$ as input.
}
\begin{thm}\label{thm:satk}
  \ptimesat
\end{thm}

  To prove Theorem~\ref{thm:satk}, we show a polynomial-time algorithm
  for computing $|\Sat(D,q,k)|$ for $q$ as in the theorem.  The
  pseudocode is depicted in Figure~\ref{fig:sat}.

  We assume in the algorithm that $D\endo$ contains only facts that
  are homomorphic images of atoms of $q$ (i.e., facts $f$ such that
  there is a mapping from an atom of $q$ to $f$). In the terminology
  of Conitzer and Sandholm~\cite{Conitzer:2004:CSV:1597148.1597185}, regarding the computation of the Shapley value,
  the function defined by $q$ \e{concerns} only the subset $C$ of
  $D\endo$ consisting of these facts (i.e., the satisfaction of $q$ by
  any subset of $D$ does not change if we intersect with $C$), and so,
  the Shapley value of every fact in $D\endo\setminus C$ is zero and
  the Shapley value of any other fact is unchanged when ignoring
  $D\endo\setminus
  C$~\cite[Lemma~4]{Conitzer:2004:CSV:1597148.1597185}. Moreover,
  these facts can be found in polynomial time.  
  
  {
\begin{figure}[t]
\begin{malgorithm}{$\algname{\SAT}(D,q,k)$}{alg:sat}
\If{$\var(q)=\emptyset$} \label{line:sat-startbase}
\If{$\exists a\in\at(q)$ s.t.~$a\not\in D$}
\State \Return $0$
\EndIf
\State $A=\at(q)\cap D\endo$
\If{$|A| = k$}
\State \Return $1$
\EndIf
\State \Return $0$ \label{line:sat-endbase}
\EndIf
\State $\result\gets 0$
\If{$q$ has a root variable}\label{line:sat-sharedstart}
\State $x\gets$ a root variable of $q$
\State $V_x\gets$ the set $\set{v_1,\dots,v_n}$ of values for $x$
\ForAll{$i \in \set{1,\dots,|V_x|}$} \label{line:dynamic-start}
\ForAll{$j \in \set{0,\dots,k}$}
\State $f_{i,j}= \leftarrow \algname{\SAT}(D^{v_i},q_{[x\rightarrow v_i]},j)$ \label{line:dynamic-firstend}
\EndFor
\EndFor
\State $P_1^\ell=f_{1,\ell}$ for all $\ell\in\set{0,\dots,k}$
\ForAll{$i \in \set{2,\dots,|V_x|}$} \label{line:sat-externalfor}
\ForAll{$\ell \in \set{0,\dots,k}$}
\State $P_i^\ell \leftarrow 0 $
\ForAll{$j \in \set{0,\dots,\ell}$}
\State $P_i^\ell \gets P_i^\ell + P_{i-1}^{\ell-j}\cdot f_{i,j}+\big[\binom{\sum_{r=1}^{i-1}|D^{v_r}\endo| }{\ell-j} - P_{i-1}^{\ell-j}\big]\cdot f_{i,j} + P_{i-1}^{\ell-j}\cdot
\big[\binom{|D^{v_i}\endo| }{j} - f_{i,j}\big]$\label{line:dynamic-end}
\EndFor
\EndFor
\EndFor 
\State $\result \leftarrow P^k_n$ \label{line:sat-sharedend}
\Else
\State let $q=q_1\land q_2$ where $\var(q_1)\cap\var(q_2)=\emptyset$\label{line:sat-disjointstart}
\State let $D^1$ and $D^2$ be the restrictions of $D$ to the relations of $q_1$ and $q_2$, respectively
\ForAll{$k_1,k_2$ s.t.~$k_1+k_2=k$}
\State $\result\gets\result + \algname{\SAT}(D^1,q_1,k_1)\cdot \algname{\SAT}(D^2,q_2,k_2)$\label{line:sat-disjointend}
\EndFor
\EndIf
\State\Return $\result$
\end{malgorithm}
\caption{\label{fig:sat} An algorithm for computing $|\Sat(D,q,k)|$ where
$q$ is a hierarchical Boolean CQ without self-joins.}
\end{figure}
}

As expected for a hierarchical query, our algorithm is a recursive
procedure that acts differently in three different cases: \e{(a)} $q$
has no variables (only constants), \e{(b)} there is a \e{root}
variable $x$, that is, $x$ occurs in all atoms of $q$, or \e{(c)} $q$
consists of two (or more) subqueries that do not share any variables.
Since $q$ is hierarchical, at least one of these cases always
applies~\cite{DBLP:journals/jacm/DalviS12}.

  In the first case
  (lines~\ref{line:sat-startbase}-\ref{line:sat-endbase}), every atom
  $a$
  of $q$ can be viewed as a fact. Clearly, if one of the
  facts in $q$ is not present in $D$, then there is no subset $E$ of
  $D\endo$ of any size such that $(D\exo\cup E)\models q$, and the
  algorithm will return $0$. Otherwise, suppose that $A$ is the set of
  endogenous facts of $q$ (and the remaining atoms of $q$, if any, are
  exogenous). Due to our assumption that every fact of $D\endo$ is a homomorphic image of an atom of $q$, the single choice of a subset of facts that makes the query true is $A$; therefore, the algorithm returns $1$ if $k=|A|$ and $0$ otherwise. 
  

  Next, we consider the case where $q$ has a root variable $x$
  (lines~\ref{line:sat-sharedstart}-\ref{line:sat-sharedend}). We
  denote by $V_x$ the set $\set{v_1,\dots,v_n}$ of values
    that $D$
  has in attributes that correspond to an occurrence of $x$.
  For example, if $q$ contains the atom $R(x,y,x)$ and $D$ contains a
  fact $R(\val{a},\val{b},\val{a})$, then $\val{a}$ is one of the
  values in $V_x$. We also denote by $q_{[x\rightarrow v_i]}$ the
  query that is obtained from $q$ by substituting $v_i$ for $x$, and
  by $D^{v_i}$ the subset of $D$ that consists of facts with the value
  $v_i$ in every attribute where $x$ occurs in $q$.

  We solve the problem for this case using a simple dynamic
  program. We denote by $P_i^\ell$ the number of subsets of size $\ell$ of
  $\bigcup_{r=1}^{i} D^{v_r}\endo$ that satisfy the query (together
  with the exogenous facts in $\bigcup_{r=1}^{i} D^{v_r}\exo$). Our
  goal is to find $P_n^k$, which is the number of subsets $E$ of size
  $k$ of $\bigcup_{r=1}^{n} D^{v_r}\endo$. Note that this union is
  precisely $D\endo$, due to our assumption that $D\endo$ contains only facts that can be obtained from atoms of $q$ via an assignment to the variables. First, we compute, for each value $v_i$,
  and for each $j\in\set{0,\dots,k}$, the number $f_{i,j}$ of subsets
  $E$ of size $j$ of $D^{v_i}\endo$ such that
  $(D^{v_i}\exo\cup E)\models q$, using a recursive call. In the
  recursive call, we replace $q$ with $q_{[x\rightarrow v_i]}$, as
  $D^{v_i}$ contains only facts that use the value $v_i$ for the
  variable $x$; hence, we can reduce the number of variables in $q$ by
  substituting $x$ with $v_i$.
  Then, for each $\ell\in\set{0,\dots,k}$ it clearly holds that
  $P_1^\ell=f_{1,\ell}$. For each $i\in\set{2,\dots,|V_x|}$ and
  $\ell \in \{0,\cdots,k\}$, we compute $P_i^\ell$ in the following
  way. Each subset $E$ of size $\ell$ of $\bigcup_{r=1}^{i} D^{v_r}\endo$
  contains a set $E_1$ of size $j$ of facts from $D^{v_i}\endo$ (for
  some $j\in\set{0,\dots,\ell}$) and a set $E_2$ of size $\ell-j$ of facts
  from $\bigcup_{r=1}^{i-1} D^{v_r}\endo$. If the subset $E$ satisfies
  the query, then precisely one of the following holds:
\begin{enumerate}
    \item $(D^{v_i}\exo\cup E_1)\models q$ and $(\bigcup_{r=1}^{i-1} D^{v_r}\exo\cup E_2)\models q$,
    \item $(D^{v_i}\exo\cup E_1)\models q$, but $(\bigcup_{r=1}^{i-1} D^{v_r}\exo\cup E_2)\not\models q$,
    \item $(D^{v_i}\exo\cup E_1)\not\models q$, but $(\bigcup_{r=1}^{i-1} D^{v_r}\exo\cup E_2)\models q$.
\end{enumerate}
Hence, we add to $P_i^\ell$ the value $P_{i-1}^{\ell-j}\cdot f_{i,j}$ that
corresponds to Case~(1), the value
$$\left(\binom{\sum_{r=1}^{i-1}|D^{v_r}\endo|}{\ell-j}-P_{i-1}^{\ell-j}\right)\cdot
f_{i,j}$$ that corresponds to Case~(2), and the value
$$P_{i-1}^{\ell-j}\cdot \left(\binom{|D^{v_i}\endo|}{j}-f_{i,j}\right)$$
that corresponds to Case~(3). Note that we have all the values
$P_{i-1}^{\ell-j}$ from the previous iteration of the for loop of
line~\ref{line:sat-externalfor}.

Finally, we consider the case where $q$ has two nonempty subqueries
$q_1$ and $q_2$ with disjoint sets of variables
(lines~\ref{line:sat-disjointstart}-\ref{line:sat-disjointend}). For
$j\in\set{1,2}$, we denote by $D^j$ the set of facts from $D$ that
appear in the relations of $q_j$. (Recall that $q$ has no self-joins;
hence, every relation can appear in either $q_1$ or $q_2$, but not in
both.) Every subset $E$ of $D$ that satisfies $q$ must contain a
subset $E_1$ of $D^1$ that satisfies $q_1$ and a subset $E_2$ of $D^2$
satisfying $q_2$. Therefore, to compute $|\Sat(D,q,k)|$, we consider
every pair $(k_1,k_2)$ of natural numbers such that $k_1+k_2=k$,
compute $|\Sat(D^1,q_1,k_1)|$ and $|\Sat(D^2,q_2,k_2)|$ via a
recursive call, and add the product of the two to the result.

{\color{black}The correctness and efficiency of $\algname{\SAT}$ is stated in the following lemma.

\begin{lem}\label{lemma:sat-correct}
 Let $q$ be a hierarchical Boolean CQ without self-joins. Then, $\algname{\SAT}(D,q,k)$ returns the number $|\Sat(D,q,k)|$
  of subsets $E$ of $D\endo$ such that $|E|=k$ and $D\exo\cup E\models q$,
  given $D$ and $k$ as input. Moreover, $\algname{\SAT}(D,q,k)$ terminates in polynomial time in $k$ and $|D|$.
\end{lem}

We have already established the correctness of the algorithm. Thus, we now consider the complexity claim of Lemma~\ref{lemma:sat-correct}. The number of recursive calls in each step is polynomial in $k$ and $|D|$. In particular, in the dynamic programming part of the algorithm (lines~\ref{line:dynamic-start}-\ref{line:dynamic-end}), we make $(k+1)\cdot |V_x|$ recursive calls. Clearly, it holds that $|V_x|\le |D|$. Furthermore, we make $2(k+1)$ recursive calls in lines~\ref{line:sat-disjointstart}-\ref{line:sat-disjointend}. Finally, in each recursive call, we reduce the number of variables in $q$ by at least one. Thus, the depth of the reduction is bounded by the number of variables in query $q$, which is a constant when considering data complexity.

{
\definecolor{Gray}{gray}{0.9}
\def\emprow{\multicolumn{3}{l}{}}
\begin{figure}[t]
  {\centering
\null\hfill
\begin{subfigure}[b]{0.1\linewidth}
\begin{tabular}{|c|c|} 
\cline{1-2}
\rowcolor{Gray}
\multicolumn{2}{l}{R}\\\cline{1-2}
$\att{A}$ & $\att{B}$ \\\cline{1-2}
$\val{1}$ & $\val{2}$\\
$\val{1}$ & $\val{3}$\\
$\val{2}$ & $\val{1}$\\
$\val{3}$ & $\val{1}$\\
\cline{1-2}
\end{tabular}
\end{subfigure}
\hfill
\begin{subfigure}[b]{0.1\linewidth}
\begin{tabular}{|c|c|} 
\cline{1-2}
\rowcolor{Gray}
\multicolumn{2}{l}{S}\\\cline{1-2}
$\att{A}$ & $\att{B}$ \\\cline{1-2}
$\val{1}$ & $\val{1}$\\
$\val{1}$ & $\val{5}$\\
$\val{2}$ & $\val{3}$\\
$\val{2}$ & $\val{4}$\\
\cline{1-2}
\end{tabular}
\end{subfigure}
\hfill
\begin{subfigure}[b]{0.1\linewidth}
\begin{tabular}{|c|c|} 
\cline{1-2}
\rowcolor{Gray}
\multicolumn{2}{l}{T}\\\cline{1-2}
$\att{A}$ & $\att{B}$\\\cline{1-2}
$\val{1}$ & $\val{1}$\\
$\val{2}$ & $\val{2}$\\
$\val{3}$ & $\val{3}$\\
$\val{5}$ & $\val{6}$\\
\cline{1-2}
\end{tabular}
\end{subfigure}
\hfill
\begin{subfigure}[b]{0.07\linewidth}
\begin{tabular}{|c|} 
\cline{1-1}
\rowcolor{Gray}
\multicolumn{1}{l}{U}\\\cline{1-1}
$\att{A}$\\\cline{1-1}
$\val{1}$\\
$\val{2}$\\
$\val{3}$\\
$\val{4}$\\
\cline{1-1}
\end{tabular}
\end{subfigure}\hfill\null}
\caption{\label{fig:sat-running} The database of Example~\ref{example:sat-running}. }
\end{figure}
}

\newcommand{\pluseq}{\mathrel{+}=}

\begin{exa}\label{example:sat-running}
We now illustrate the execution of $\algname{\SAT}(D,q,k)$ on the database $D$ of Figure~\ref{fig:sat-running}, the query $q()\dl R(x,y),S(x,z),T(w,w),U(w)$ and $k=4$. We assume that all facts in $D$ are endogenous. Since $q$ does not have a root variable, the condition of line~\ref{line:sat-sharedstart} does not hold. Hence, we start by considering the two disjoint sub-queries $q_1()\dl R(x,y),S(x,z)$ and $q_2()\dl T(w,w),U(w)$ in line~\ref{line:sat-disjointstart}, and the corresponding databases $D_1$ that contains the relations $R$ and $S$ and $D_2$ that contains the relations $T$ and $U$. Note that $q_1$ and $q_2$ indeed do not share any variables.

Each set of facts that satisfies $q$ contains four facts of the form $R(a,b)$, $S(a,c)$, $T(d,d)$ and $U(d)$ for some values $a,b,c,d$. Clearly, it holds that $\set{R(a,b),S(a,c)}\models q_1$ and $\set{T(d,d),U(d)}\models q_2$; thus, we compute $\algname{\SAT}(D,q,4)$ using $10$ (that is, $2(k+1)$) recursive calls to $\algname{\SAT}$.
\begin{align*}
\algname{\SAT}(D,q,4) &= \algname{\SAT}(D_1,q_1,0)\cdot \algname{\SAT}(D_2,q_2,4)\\
&+\algname{\SAT}(D_1,q_1,1)\cdot \algname{\SAT}(D_2,q_2,3)\\
&+\algname{\SAT}(D_1,q_1,2)\cdot \algname{\SAT}(D_2,q_2,2)\\
&+\algname{\SAT}(D_1,q_1,3)\cdot \algname{\SAT}(D_2,q_2,1)\\
&+\algname{\SAT}(D_1,q_1,4)\cdot \algname{\SAT}(D_2,q_2,0)
\end{align*}

Now, $q_1$ contains a root variable $x$; thus, in each recursive call with the query $q_1$, the condition of line~\ref{line:sat-sharedstart} holds. We will illustrate the execution of this part of the algorithm using $\algname{\SAT}(D_1,q_1,3)$. Note that in a homomorphism from $R(x,y)$ to $D_1$, the variable $x$ is mapped to one of three values, namely $\val{1}$, $\val{2}$, or $\val{3}$. Similarly, in a homomorphism from $S(x,z)$ to $D_1$, the value $x$ is mapped to either $\val{1}$ or $\val{2}$. Hence, it holds that $V_x=\set{\val{1},\val{2},\val{3}}$.

For each value $a_i$ in $V_x$ (where $a_1=\val{1}, a_2=\val{2}, a_3=\val{3}$), we consider the query $q_{[x\rightarrow a_i]}$ which is $R(a_i,y),S(a_i,z)$, and the database $D^{a_i}$ containing the facts that use the value $a_i$ for the variable $x$. That is, the database $D^{\val{1}}$ contains the facts $\set{R(\val{1},\val{2}), R(\val{1},\val{3}), S(\val{1},\val{1}), S(\val{1},\val{5})}$, the database $D^{\val{2}}$ contains the facts $\set{R(\val{2},\val{1}), S(\val{2},\val{3}), S(\val{2},\val{4})}$, and the database $D^{\val{3}}$ contains the fact $\set{R(\val{3},\val{1})}$. Then, for each one of the three values, and for each $j\in\set{0,\dots,3}$, we compute the number $f_{i,j}$ of subsets of size $j$ of $D^{a_i}$ that satisfy $q$, using the recursive call $\algname{\SAT}(D^{a_i},q_{[x\rightarrow a_i]},j)$. The reader can easily verify that the following holds.
\begin{align*}
    \centering
    & f_{1,0}=0 \quad f_{1,1}=0 \quad f_{1,2}=4 \quad f_{1,3}=4\\
    & f_{2,0}=0 \quad f_{2,1}=0 \quad f_{2,2}=2 \quad f_{2,3}=1\\
    & f_{3,0}=0 \quad f_{3,1}=0 \quad f_{3,2}=0 \quad f_{3,3}=0
\end{align*}
Next, we compute, for each $i\in\set{1,\dots,3}$ and $l\in\set{0,\dots,3}$, the number $P_i^l$ of subsets of size $l$ of $\bigcup_{r=1}^i D^{a_r}$ that satisfy $q$. We begin with $a_i$ (i.e., the value $\val{1}$), in which case it holds that $P_1^l=f_{1,l}$. Hence, we have that:
\begin{align*}
    \centering
    & P_1^0=P_1^1=0 \quad\quad P_1^2=P_1^3=4
\end{align*}
 Next, for each $l\in\set{0,\dots,3}$, we compute the number $P_2^l$ of subsets of $D^{\val{1}}\cup D^{\val{2}}$ that satisfy $q$. Each such subset contains $j$ facts from $D^{\val{2}}$ and $l-j$ facts for $D^{\val{1}}$ for some $j\in\set{0,\dots,l}$. Recall that $D^{\val{1}}$ contains four facts and $D^{\val{2}}$ contains three facts. Hence, we have the following computations for $l=0$.
 \begin{align*}
    \centering
    & P_2^0 = 0\\
    & P_2^0 \pluseq P_1^0\cdot f_{2,0} + \left[\binom{4}{0}-P_1^0\right]\cdot f_{2,0} + P_1^0\cdot \left[\binom{3}{0}-f_{2,0}\right]=0+0\cdot 0+1\cdot 0 +0\cdot 1=0
\end{align*}
In the first line we initialize $P_2^0$. Then, in the second line, we consider $j=0$, which is the only possible $j$ in this case. Next, for $l=1$, we compute the following.
 \begin{align*}
    \centering
    & P_2^1 = 0\\
    & P_2^1 \pluseq P_1^1\cdot f_{2,0} + \left[\binom{4}{1}-P_1^1\right]\cdot f_{2,0} + P_1^1\cdot \left[\binom{3}{0}-f_{2,0}\right]=0+0\cdot 0+4\cdot 0 +0\cdot 1=0\\
    & P_2^1 \pluseq P_1^0\cdot f_{2,1} + \left[\binom{4}{0}-P_1^0\right]\cdot f_{2,1} + P_0^1\cdot \left[\binom{3}{1}-f_{2,1}\right]=0+0\cdot 0+1\cdot 0 +0\cdot 3=0
\end{align*}
Here, in the second line, we consider $j=0$ (i.e., choosing zero facts from $D^{\val{2}}$ and one fact from $D^{\val{1}}$), and in the third line we consider $j=1$ (i.e., choosing one fact from $D^{\val{2}}$ and zero facts from $D^{\val{1}}$). Next, we have $l=2$.
 \begin{align*}
    \centering
    & P_2^2 = 0\\
    & P_2^2 \pluseq P_1^2\cdot f_{2,0} + \left[\binom{4}{2}-P_1^2\right]\cdot f_{2,0} + P_1^2\cdot \left[\binom{3}{0}-f_{2,0}\right]=0+4\cdot 0+2\cdot 0 +4\cdot 1=4\\
    & P_2^2 \pluseq P_2^1\cdot f_{2,1} + \left[\binom{4}{1}-P_1^1\right]\cdot f_{2,1} + P_1^1\cdot \left[\binom{3}{1}-f_{2,1}\right]=4+0\cdot 0+4\cdot 0 +0\cdot 3=4\\
    &P_2^2 \pluseq P_2^0\cdot f_{2,2} + \left[\binom{4}{0}-P_1^0\right]\cdot f_{2,2} + P_1^0\cdot \left[\binom{3}{2}-f_{2,2}\right]=4+0\cdot 2+1\cdot 2 +0\cdot 1=6
\end{align*}
Finally, we consider $l=3$.
 \begin{align*}
    \centering
    & P_2^3 = 0\\
    & P_2^3 \pluseq P_1^3\cdot f_{2,0} + \left[\binom{4}{3}-P_1^3\right]\cdot f_{2,0} + P_1^3\cdot \left[\binom{3}{0}-f_{2,0}\right]=0+4\cdot 0+0\cdot 0 +4\cdot 1=4\\
    & P_2^3 \pluseq P_1^2\cdot f_{2,1} + \left[\binom{4}{2}-P_1^2\right]\cdot f_{2,1} + P_1^2\cdot \left[\binom{3}{1}-f_{2,1}\right]=4+4\cdot 0+2\cdot 0 +4\cdot 3=16\\
    & P_2^3 \pluseq P_1^1\cdot f_{2,2} + \left[\binom{4}{1}-P_1^1\right]\cdot f_{2,2} + P_1^1\cdot \left[\binom{3}{2}-f_{2,2}\right]=16+0\cdot 2+4\cdot 2 +0\cdot 1=24\\
    & P_2^3 \pluseq P_1^0\cdot f_{2,3} + \left[\binom{4}{0}-P_1^0\right]\cdot f_{2,3} + P_1^0\cdot \left[\binom{3}{3}-f_{2,3}\right]=24+0\cdot 1+1\cdot 1 +0\cdot 0=25
\end{align*}
We conclude that:
\begin{align*}
    \centering
    & P_2^0=P_2^1=0 \quad\quad P_2^2=6 \quad\quad P_2^3=25
\end{align*}

Now, we can compute $P_3^l$ for each $l\in\set{0,\dots,3}$ in a similar way, using the above values and the values $f_{3,j}$ that we have computed before. We omit the computations here. The final results are the following.

\begin{align*}
    \centering
    & P_3^0=P_3^1=0 \quad\quad P_3^2=6 \quad\quad P_2^3=31
\end{align*}

Then, $\algname{\SAT}(D_1,q_1,3)$ returns $P_2^3$ which is the number of subset of size $3$ of $D_1$ that satisfy the query.

Finally, we illustrate the base case of the algorithm (that is, lines~\ref{line:sat-startbase}-\ref{line:sat-endbase}). To do that, we use the recursive call $\algname{\SAT}(D_2,q_2,3)$ from the first step of the execution. Recall that $q_2()\dl T(w,w), U(w)$ and $D_2$ contains all the facts in $T$ and $U$. The query $q_2$ contains a single variable $w$. In a homomorphism from $T(w,w)$ to $D_2$, this variable is mapped to one of three values, namely $\val{1}$, $\val{2}$, or $\val{3}$. Note that there is no homomorphism from $T(w,w)$ to the fact $T(\val{5},\val{6})$; hence, the values $\val{5}$ and $\val{6}$ are not in $V_w$. In addition, in a homomorphism from $U(w)$ to $D_2$, the variable $w$ is mapped to one of $\val{1}$, $\val{2}$, $\val{3}$, or $\val{4}$; thus, $V_w=\set{\val{1},\val{2},\val{3},\val{4}}$.

In every recursive call, we will substitute one of the values in $V_w$ for $w$. One of the recursive calls will be $\algname{\SAT}(D_2^{\val{1}},q_2',2)$, where $q_2'()\dl T(\val{1},\val{1}),U(\val{1})$. Here, $D_2^{\val{1}}$ contains every atom of $q$, and $k=|A|$; hence, the recursive call will return $1$. On the other hand, the result of $\algname{\SAT}(D_2^{\val{1}},q_2',3)$ will be zero; as there are only two facts in $D_2^\val{1}$, while $k=3$. The result of $\algname{\SAT}(D_2^{\val{1}},q_2',1)$ will also be zero, since in this case $k=1$ and $|A|=2$; thus, $k<|A|$.
Finally, for the recursive call $\algname{\SAT}(D_2^{\val{4}},q_2'',2)$, where $q_2''()\dl T(\val{4},\val{4}),U(\val{4})$, the result will be zero, as the fact $T(\val{4},\val{4})$ is not in the database.\qed
\end{exa}
}

\subsubsection{Hardness side} We now give the proof of the hardness
side of Theorem~\ref{thm:bcq-dichotomy}. Membership in
$\fpsharpp$ is straightforward since, as aforementioned in
Equation~\eqref{eq:shapley-bcq}, the Shapley value can be computed in
polynomial time given an oracle to the problem of counting the number
of subsets $E\subseteq D\endo$ of size $k$ such that
$(D\exo\cup E)\models q$, and this problem is in $\sharpp$.
Similarly to Dalvi and
Suciu~\cite{DBLP:conf/vldb/DalviS04}, our proof of hardness consists
of two steps. First, we prove the $\fpsharpp$-hardness of computing
$\shapley(D,\cqrst,f)$, where $\cqrst$ is given in~\eqref{eq:cqrst}.
Second, we reduce the computation of $\shapley(D,\cqrst,f)$ to the
problem of computing $\shapley(D,q,f)$ for any non-hierarchical CQ $q$
without self-joins. The second step is the same as that of Dalvi and
Suciu~\cite{DBLP:conf/vldb/DalviS04}, and we will give the proof here for completeness. The proof of the first step---hardness of
computing $\shapley(D,\cqrst,f)$ (stated by Lemma~\ref{lemma:hardnessbipartite}), is considerably more involved than the corresponding proof
of Dalvi and Suciu~\cite{DBLP:conf/vldb/DalviS04} that computing the
probability of $\cqrst$ in a tuple-independent probabilistic database
(TID) is $\fpsharpp$-hard. This is due to the coefficients of the Shapley value that do not seem to easily factor out.

\def\hardnessqrst{
Computing $\shapley(D,\cqrst,f)$ is $\fpsharpp$-complete.
}

\def\bisat{\mathsf{\#biSAT}}

\begin{prop}\label{lemma:hardnessbipartite}
\hardnessqrst
\end{prop}
\begin{proof}
{\color{black}The proof is by a (Turing)
reduction from the problem of computing the number $|\is(g)|$ of
independent sets of a given bipartite graph $g$, which is the same
(via immediate reductions) as the problem of computing the number of
satisfying assignments of a bipartite monotone $2$-DNF formula, which
we denote by $\bisat$.  Dalvi and Suciu~\cite{DBLP:conf/vldb/DalviS04}
also proved the hardness of $\cqrst$ (for the problem of query
evaluation over TIDs) by reduction from $\bisat$. Their reduction is a
simple construction of a single input database, followed by a
multiplication of the query probability by a number. It is not at all
clear to us how such an approach can work in our case and, indeed, our
proof is more involved.  Our reduction takes the general approach that
Dalvi and Suciu~\cite{DBLP:journals/jacm/DalviS12} used (in a
different work) for proving that the CQ $q()\dl R(x,y),R(y,z)$ is hard
over TIDs: solve several instances of the problem for the construction
of a full-rank set of linear equations. The problem itself, however,
is quite different from ours.  This general technique has also been
used by Aziz et al.~\cite{DBLP:conf/stacs/AzizK14} for proving the
hardness of computing the Shapley value for a \e{matching game} on
unweighted graphs, which is again quite different from our problem.

\begin{figure}
\centering
\input{bipartite.pspdftex}
\caption{Constructions in the reduction of the proof of
  Lemma~\ref{lemma:hardnessbipartite}. Relations $R/1$ and $T/1$
  consist of endogenous facts and $S/2$ consists of exogenous facts.}
\label{fig:reduction}
\end{figure}

In more detail, the idea is as follows. Given an input bipartite graph
$g=(V,E)$ for which we wish to compute $|\is(g)|$, we construct $n+2$
different input instances $(D_j,f)$, for $j=0,\dots,n+1$, of the problem of computing
$\shapley(D_j,\cqrst,f)$, where $n=|V|$. Each instance provides us with
an equation over the numbers $|\is(g,k)|$ of independent sets of size
$k$ in $g$ for $k=0,\dots,n$. We then show that the set of equations
constitutes a non-singular matrix that, in turn, allows us to extract
the $|\is(g,k)|$ in polynomial time (e.g., via Gaussian
elimination). This is enough, since
$|\is(g)|=\sum_{k=0}^{n}|\is(g,k)|$.

Our reduction is illustrated in Figure~\ref{fig:reduction}.  Given the
graph $g$ (depicted in the leftmost part), we construct $n+2$ graphs
by adding new vertices and edges to $g$. For each such graph, we build
a database that contains an endogenous fact $R(v)$ for every left
vertex $v$, an endogenous fact $T(u)$ for every right vertex $u$, and an
exogenous fact $S(v,u)$ for every edge $(v,u)$.  In each constructed database
$D_j$, the fact $f=R(\val{0})$ represents a new left node, and we compute
$\shapley(D_j,\cqrst,f)$.  In $D_0$, the node of $f$ is connected to
every right vertex (i.e., we add an exogenous fact $S(\val{0},u)$ for every right vertex $u$). We use this database to compute a specific value (from the Shapley value of $f$), as we explain next. 

Instead of directly computing the Shapley value of $f$, we compute the complement of the Shapley value. To do that, we consider the permutations $\sigma$ where $f$ does not affect the query result. This holds in one of two cases:
\begin{enumerate}
    \item No fact of $T$ appears before $f$ in $\sigma$,
    \item At least one pair $\set{R(v),T(u)}$ of facts, such that there is a fact $S(u,v)$ in $D_0$, appears in $\sigma$ before $f$.
\end{enumerate}
The number of permutations where the first case holds is:
$$P_0^1=\frac{(n+1)!}{n_T+1}$$
where $n_T$ is the number of vertices on the right-hand side of the graph $g$ (namely, the number of facts in $T$). This holds since each one of the facts of $T$ and the fact $f$ have an equal chance to be selected first (among these facts) in a random permutation. We are looking for the permutations where $f$ is chosen before any fact of $T$; hence, we are looking at $1/(n_T+1)$ of all permutations.

Now, we compute the number of permutations $\sigma$ where the second case holds. To do that, we have to count the permutations $\sigma$ where $\sigma_f$ corresponds to a set of vertices from $g$ (the original graph) that is not an independent set. Let us denote by $|\ds(g,k)|$ the number of subsets of vertices of size $k$ from $g$ that are not independent sets. Then, the number of permutations satisfying the above is:
$$P_0^2=\sum_{k=2}^n |\ds(g,k)|\cdot k!\cdot (n-k)!$$
Recall that the fact $f$ corresponds to a new vertex that does not occur in the original bipartite graph $g$; hence, for each $k$, we have $k$ facts that appear before $f$ in the permutation and $n+1-k-1=n-k$ facts that appear after $f$.
We can now express the Shapley value of $f$ in terms of $P_0^1$ and $P_0^2$:
$$\shapley(D_0,\cqrst,f)=1-\frac{P_0^1+P_0^2}{(n+1)!}$$ 
Then, the value $P_0^2$ can be computed from $\shapley(D_0,\cqrst,f)$ using the following formula.
$$P_0^2=(1-\shapley(D_0,\cqrst,f))\cdot(n+1)!-P_0^1$$
We will use this value later in our proof.

Next, for $j=1,\dots,n+1$, we construct a
database $D_j$ that is obtained from $g$ by adding $f=R(\val{0})$ and facts $T(0_1),\dots,T(0_j)$ of $j$ new
right nodes, all connected to $f$ (by adding an exogenous fact $S(\val{0},0_i)$ for each $i\in\set{1,\dots,j})$.
We again compute the complement of the Shapley value of $f$, for each database $D_j$. The permutations where $f$ does not affect the query result are those satisfying one of two properties:
\begin{enumerate}
    \item At least one pair $\set{R(v),T(u)}$ of facts for $u,v\in V$, such that there exists a fact $S(v,u)$ in $D_j$, appears in $\sigma$ before $f$,
    \item No pair $\set{R(v),T(u)}$ of facts for $u,v\in V$, such that there exists a fact $S(v,u)$ in $D_j$, as well as none of the facts $T(0_1),\dots,T(0_j)$ appear in $\sigma$ before $f$.
\end{enumerate}
Recall that $V$ is the set of vertices of the original bipartite graph $g$.

Note that if the first condition holds, then it does not matter if we choose a fact from the set $\set{T(0_1),\dots,T(0_j)}$ before choosing $f$ or not (that is, the fact $f$ will not affect the query result regardless of the positions of these facts). Hence, we first ignore these facts, and compute the number of permutations of the rest of the facts that satisfy the first condition:
\[\sum_{k=2}^n |\ds(g,k)|\cdot k!\cdot (n-k)!\]
From each such permutation $\sigma$, we can then generate $m_j$ permutations of all the $n+j+1$ facts in $D_j$ by considering all the $m_j$ possibilities to add the facts of $\set{T(0_1),\dots,T(0_j)}$ to the permutation. Note that this is the same $m_j$ for each permutation, and it holds that $m_j=\binom{n+j+1}{j}\cdot j!$ (i.e., we select $j$ positions for the facts of $\set{T(0_1),\dots,T(0_j)}$ and place them in these position in one of $j!$ possible permutations, while placing the rest of the facts in the remaining positions in the order defined by $\sigma$). Moreover, using this procedure we cover all the permutations of the facts in $D_j$ that satisfy the first condition, since for each one of them there is a single corresponding permutation of the facts in $D_j\setminus\set{T(0_1),\dots,T(0_j)}$. Hence, the number of permutations of the facts in $D_j$ that satisfy the first property is
\[m_j\cdot\sum_{k=2}^n |\ds(g,k)|\cdot k!\cdot (n-k)!=m_j\cdot P_0^2\]
Recall that we have seen earlier that the value $P_0^2$ can be computed from $\shapley(D_0,\cqrst,f)$.

Next, we compute the number of permutations that satisfy the second property:
$$P_j=\sum_{k=0}^n |\is(g,k)|\cdot k!(n+j-k)!$$
This holds since each permutation $\sigma$ where $\sigma_f$ does not contain any fact $T(0_j)$ and any pair $\set{R(v),T(u)}$ of facts such that there is a fact $S(u,v)$ in $D_j$, corresponds to an independent set of $g$.
Hence, for each $j=1,\dots,n+1$ we get an equation of the form:
$$\shapley(D_j,\cqrst,f)=1-\frac{m_j\cdot P_0^2+P_j}{(n+j+1)!}$$
And we can compute $P_j$ from $\shapley(D_j,\cqrst,f)$ in the following way.
$$P_j=(n+j+1)!\big(1-\shapley(D_j,\cqrst,f)\big)-m_j\cdot P_0^2$$

From these equations we now
extract a system $Ax=y$ of $n+1$ equations over $n+1$ variables (i.e.,
$|\is(g,0)|,\dots,|\is(g,n)|$), where each $S_j$ stands for
\revthree{$1-\shapley(D_j,\cqrst,f)$}.
\revone{
\begin{gather*}
  \left(\!\! {\begin{array}{cccc}
        0!(n+1)! & 1!n! & \dots & n!1! \\
        0!(n+2)! & 1!(n+1)! & \dots & n!2! \\
        \vdots & \vdots & \vdots & \vdots \\
        0!(2n+1)! & 1!(2n)! & \dots & n!(n+1)!
      \end{array} }\!\! \right)\!\!\!
  \left( {\begin{array}{c}
  |\is(g,0)| \\
   |\is(g,1)| \\
   \vdots \\
   |\is(g,n)|
 \end{array} } \right)\!
= \!\left(\!\! {\begin{array}{c}
      (n+2)!S_1-m_1\cdot P_0^2 \\
      (n+3)!S_2-m_2\cdot P_0^2 \\
      \vdots \\
      (2n+2)!S_{n+1}-m_{n+1}\cdot P_0^2
  \end{array} }\!\! \right)\!
\end{gather*}}
By an elementary algebraic manipulation of $A$ (i.e., dividing each column $j$ by the constant $j!$ and reversing the order of the columns), we obtain the matrix
with the coefficients $a_{i,j}=(i+j+1)!$ that
Bacher~\cite{determinants2002} proved to be non-singular (and, in
fact, that $\prod_{i=0}^{n-1} i!(i+1)!$ is its determinant).  We then
solve the system as discussed earlier to obtain $|\is(g,k)|$ for each $k$, and, consequently, compute the value $\is(g)=\sum_{k=0}^n \is(g,k)$.
}
\end{proof}

{\color{black}
Finally, we show that computing $\shapley(D,q,f)$ is hard for any non-hierarchical Boolean CQ $q$ without self-joins, by constructing a reduction from the problem of computing $\shapley(D,\cqrst,f)$. As aforementioned, our reduction is very similar to the corresponding reduction of Dalvi and Suciu~\cite{DBLP:conf/vldb/DalviS04}, and we give it here for completeness. We will also use this result in Section~\ref{sec:measures}.

\begin{lemma}\label{lemma:reduction-from-cqrst}
Let $q$ be a non-hierarchical Boolean CQ without self-joins. Then,
computing $\shapley(D,q,f)$ is $\fpsharpp$-complete
\end{lemma}
\begin{proof}
We build a reduction from the problem of computing $\shapley(D,\cqrst,f)$ to the problem of computing $\shapley(D,q,f)$. Since $q$ is not hierarchical, there exist two variables $x,y\in \var(q)$, such that $A_x \cap A_y \neq \emptyset$, while $A_x\not\subseteq A_y$ and $A_y\not\subseteq A_x$; hence, we can choose three atoms $\alpha_x,\alpha_y$ and $\alpha_{(x,y)}$ in $q$ such that:
\begin{itemize}
\item $x\in \var(\alpha_x)$ and $y\notin \var(\alpha_x)$
\item $y\in \var(\alpha_y)$ and $x\notin \var(\alpha_y)$
\item $x,y\in \var(\alpha_{x,y})$
\end{itemize}
Recall that $\var(\alpha)$ is the set of variables that appear in the atom $\alpha$.

Given an input database $D$ to the first problem, we build an input database $D'$ to our problem in the following way. Let $\val{c}$ be an arbitrary constant that does not occur in $D$. For each fact $R(\val{a})$ and for each atom $\alpha\in A_x\setminus A_y$, we generate a fact $f$ over the relation corresponding to $\alpha$ by assigning the value $\val{a}$ to the variable $x$ and the value $\val{c}$ to the rest of the variables in $\alpha$. We then add the corresponding facts to $D'$. We define each new fact in the relation of $\alpha_x$ to be endogenous if and only if the original fact from $R$ is endogenous, and we define the rest of the facts to be exogenous.

Similarly, for each fact $T(\val{b})$ and for each atom $\alpha\in A_y\setminus A_x$, we generate a fact $f$ over the relation corresponding to $\alpha$ by assigning the value $\val{b}$ to the variable $y$ and the value $\val{c}$ to the rest of the variables in $\alpha$. Moreover, for each fact $S(\val{a},\val{b})$ and for each atom $\alpha\in A_x\cap A_y$, we generate a fact $f$ over the relation corresponding to $\alpha$ by assigning the value $\val{a}$ to $x$, the value $\val{b}$ to $y$ and the value $\val{c}$ to the rest of the variables in $\alpha$. In both cases, we define the new facts in $\alpha_x$ and $\alpha_{x,y}$ to be endogenous if and only if the original fact is endogenous, and we define the rest of the facts to be exogenous. Finally, for each atom $\alpha$ in $q$ that does not use the variables $x$ and $y$ (that is, $\alpha\not\in A_x\cup A_y$), we add a single exogenous fact $R_\alpha(\val{c},\dots,\val{c})$ to the relation $R_\alpha$ corresponding to $\alpha$.

We will now show that the Shapley value of each fact $R(\val{a})$ in $D$ w.r.t~$q_\rst$ is equal to the Shapley value of the corresponding fact $f$ over the relation of $\alpha_x$ in $D'$ (i.e., the fact in the relation of $\alpha_x$ that has been generated using the value $\val{a}$ that occurs in $R(\val{a})$). The same holds for a fact $T(\val{b})$ and its corresponding fact in the relation of $\alpha_y$ in $D'$, and for a fact $S(\val{a},\val{b})$ and its corresponding fact in the relation of $\alpha_{x,y}$ in $D'$. 

By definition, the Shapley value of a fact $f$ is the probability to select a random permutation $\sigma$ in which the addition of the fact $f$ changes the query result from $0$ to $1$ (i.e., $f$ is a counterfactual cause for the query w.r.t.~$\sigma_f\cup D\exo$). 
From the construction of $D'$, it holds that the number of endogenous facts in $D$ is the same as the number of endogenous facts in $D'$; hence, the total number of permutations of the facts in $D$ is the same as the total number of permutations of the facts in $D'$. It is left to show that the number of permutations of the facts in $D$ that satisfy the above condition is the same as the number of permutations of the facts in $D'$ that satisfy the above condition w.r.t.~the corresponding fact $f'$.

From the construction of $D'$ it is straightforward that a subset $E$ of $D\endo$ is such that $E\cup D\exo\models\cqrst$ if and only if the subset $E'$ of $D'\endo$ that contains for each fact $f\in E$ the corresponding fact $f'\in D'$ is such that $E'\cup D'\exo\models q$. Therefore, it also holds that if a fact $f$ is a counterfactual cause for $\cqrst$ w.r.t.~$E\cup D\exo$, the corresponding fact $f'$ is a counterfactual cause for $q$ w.r.t.~$E'\cup D'\exo$.
Thus, the number of permutations of the endogenous facts in $D$ in which $f$ affects the result of $\cqrst$ is equal to the number of permutations of the endogenous facts in $D'$ in which $f'$ changes the result of $q$. As aforementioned, the total number of permutations is the same for both $D$ and $D'$, and we conclude, from the definition of the Shapley value, that $\shapley(D,\cqrst,f)=\shapley(D',q,f')$.
\end{proof}
}

\subsection{Aggregate Functions over Conjunctive Queries}\label{sec:agg}
Next, we study the complexity of aggregate-relational queries, where
the internal relational query is a CQ. We begin with hardness. The
following theorem generalizes the hardness side of
Theorem~\ref{thm:bcq-dichotomy} and states that it is
$\fpsharpp$-complete to compute $\shapley(D,\alpha,f)$ whenever
$\alpha$ is of the form $\gamma\bracs{q}$, as defined in
Section~\ref{sec:preliminaries}, and $q$ is a non-hierarchical CQ
without self-joins.  The only exception is when $\alpha$ is a
\e{constant} numerical query (i.e., $\alpha(D)=\alpha(D')$ for
all databases $D$ and $D'$); in that case, $\shapley(D,\alpha,f)=0$
always holds.

\def\agghard{ Let $\alpha=\gamma\bracs{q}$ be a fixed
  aggregate-relational query where $q$ is a non-hierarchical CQ
  without self-joins. Computing $\shapley(D,\alpha,f)$, given $D$ and
  $f$ as input, is $\fpsharpp$-complete, unless $\alpha$ is constant.}
\begin{thm}\label{thm:agg-hard}
\agghard
\end{thm}
{\color{black}
\begin{proof}
Since $\alpha$ is not a constant function, there exists a database $\widetilde{D}$, such that $\alpha(\widetilde{D})\neq\alpha(\emptyset)$. Let $\widetilde{D}$ be a minimal such database; that is, for every database $D$ such that $q(D)\subset q(\widetilde{D})$ it holds that $\alpha(D)=\alpha(\emptyset)$. Let $q(\widetilde{D})=\set{\tup a_1,\dots,\tup a_n}$. We replace the head variables in $q$ with the corresponding constants from the answer $\tup a_1$. We denote the result by $q'$. 

We start with the following observation. The query $q'$ is a non-hierarchical Boolean CQ (recall that the definition of hierarchical queries considers only the existential variables, which are left intact in $q'$). We can break the query $q'$ into connected components $q_1',...,q_m'$, such that $\var(q_i')\cap \var(q_j')=\emptyset$ for all $i\neq j$ (it may be the case that there is only one connected component). Since $q'$ is not hierarchical, we have that $q_i'$ is not hierarchical for at least one $i\in\set{1,\dots,m}$. We assume, without loss of generality, that $q_1'$ is not hierarchical. Then, Theorem~\ref{thm:bcq-dichotomy} implies that computing $\shapley(D,q_1',f)$ is $\fpsharpp$-complete. Therefore, we construct a reduction from the problem of computing $\shapley(D,q_1',f)$ to the problem of computing $\shapley(D',\alpha,f)$.

Let $x_1,\dots,x_k$ be the head variables of $q$. If a tuple $\tup a=(v_1,\dots,v_k)$ is in $q(D)$ for some database $D$, then for every connected component $q_i$ of $q$, there is a homomorphism from $q_i$ to $D$, such that each head variable $x_j$ is mapped to the corresponding value $v_j$ from $\tup a$. On the other hand, if this does not hold for at least one of the connected components, then $(v_1,\dots,v_k)$ is not in $q(D)$.

Given an input database $D$ to the first problem, we build an input database $D'$ to our problem, as we explain next. \revone{As in the proof of Theorem~\ref{thm:satk}, we assume, without loss of generality, that $D$ contains only facts $f$ such that there is a homomorphism from an atom of $q_1'$ to $f$.}
To construct $D'$, we first add a subset of the facts of $\widetilde{D}$ to $D\exo'$ (recall that $\widetilde{D}$ is a minimal database satisfying $\alpha(\widetilde{D})\neq\alpha(\emptyset)$). For each relation $R$ that occurs in $q_i'$ for $i=\set{2,\dots,m}$, we copy all the facts from $R^{\widetilde{D}}$ to $R^{D\exo'}$. As explained above, for each answer $\tup a_i=(v_1,\dots,v_k)$ in $q(\widetilde{D})$ and for each connected component $q_i$, there is a homomorphism from $q_i$ to $\widetilde{D}$ such that each head variable $x_j$ that appears in $q_i$ is mapped to the value $v_j$; hence, the same holds for the database $D'$ and every connected component in $\set{q_2,\dots,q_m}$. Therefore, in order to have all the tuples $\set{\tup a_1,\dots,\tup a_n}$ in $q(D')$, we only need to add additional facts to the relations that appear in $q_1$ to satisfy this connected component.

Now, let $x_{j_1},\dots,x_{j_r}$ be the head variables that appear in $q_1$. For each tuple $\tup a_i$ that does not agree with $\tup a_1$ on the values of these variables (i.e., the value of at least one $x_{j_k}$ is different in $\tup a_1$ and $\tup a_i$), we generate a set of exogenous facts as follows. Assume that $\tup a_i$ uses the value $v_{j_k}$ for the head variable $x_{j_k}$. We replace each variable $x_{j_k}$ in $q_1$ with the value $v_{j_k}$. Then, we assign a new distinct value to each one of the existential variables of $q_1$. We then add the corresponding facts to $D\exo'$ (e.g., if $q_1$ now contains the atom $R(\val{a},\val{b},\val{c})$, then we add the fact $R(\val{a},\val{b},\val{c})$ to $D\exo'$). At this point, it is rather straightforward that each $\tup a_i$ that does not agree with $\tup a_1$ on the values of the head variables of $q_1$ appears in $q(D')$; however, $\tup a_1$ and each tuple $\tup a_i$ that uses the same values as $\tup a_1$ for the head variables of $q_1$ are not yet in $q(D')$. Since we assumed that $\widetilde{D}$ is minimal, we know that $\alpha(D\exo')=\alpha(\emptyset)$.

Next, we add all the facts of $D$ to $D'$. Each fact of $D\exo$ is added to $D\exo'$, and each fact of $D\endo$ is added to $D\endo'$. We prove that the following holds:
  \[\shapley(D,q_1',f)=\frac{\shapley(D',\alpha,f)}{\alpha(\widetilde{D})-\alpha(\emptyset)}\]

Let $A=\set{\tup a_{k_1},\dots,\tup a_{k_t}}$ be the set of answers that do not agree with $\tup a_1$ on the values of the head variables $x_{j_1},\dots,x_{j_r}$ of $q_1$. As explained above, we have that $A\subseteq q(D\exo')$. Moreover, since $q_1'$ was obtained from $q_1$ by replacing the head variables with the corresponding values from $\tup a_1$, and since we removed from $D$ every fact that does not agree with $q_1'$ on those values, the only possible answer of $q_1$ on $D$ is $(v_{j_1},\dots,v_{j_r})$ where $v_{j_i}$ is the value in $\tup a_1$ corresponding to the head variable $x_{j_i}$ of $q_1$. Since all the answers in $q(\widetilde{D})\setminus A$ agree with $\tup a_1$ on the values of all these variables, we have that in each permutation $\sigma$ of the endogenous facts in $D'$, one of the following holds:
\begin{enumerate} 
    \item $q(\sigma_f\cup D\exo')=A$ and $q(\sigma_f\cup D\exo'\cup\set{f})=A$,
    \item $q(\sigma_f\cup D\exo')=A$ and $q(\sigma_f\cup D\exo'\cup\set{f})=q(\widetilde{D})$.
\end{enumerate}
Clearly, the contribution of each permutation that satisfies the first condition to the Shapley value of $f$ is zero (as we assumed that $\alpha(A)=\alpha(\emptyset)$ for every $A\subset \set{\tup a_1,\dots, \tup a_n}$), while the contribution of each permutation that satisfies the second condition to the Shapley value of $f$ is $\alpha(\widetilde{D})-\alpha(\emptyset)$ (as we assumed that $\alpha(\widetilde{D})\neq\alpha(\emptyset)$).

Let $X_f$ be a random variable that gets the value $1$ if $f$ adds the answer $\tup a_1$ to the result of $q$ in the permutation and $0$ otherwise. Due to the aforementioned observations and by the definition of the Shapley value, the following holds:
  \[\shapley(D',\alpha,f)=(\alpha(\widetilde{D})-\alpha(\emptyset))\cdot
  E(X_f)\]
Note that $D$ and $D'$ contain the same endogenous facts. Moreover, the fact $f$ adds the answer $\tup a_1$ to the result of $q$ in a permutation $\sigma$ of the endogenous facts in $D'$ if and only if $f$ changes the result of $q_1'$ from $0$ to $1$ in the same permutation $\sigma$ of the endogenous facts of $D$. This holds since $f$ changes the result of $q_1'$ in $\sigma$ if and only if there exist a set of facts in $\sigma_f\cup D\exo\cup{f}$ (that contains $f$) that satisfies $q_1'$, which, as explained above, happens if and only if the same set of facts adds the answer $\tup a_1$ to the result of $q$ in $\sigma$. Therefore, the Shapley value of a fact $f$ in $D$ is:
  \[\shapley(D,q_1',f)= E(X_f)\]
where $X_f$ is the same random variable that we introduced above, and that concludes our proof.
\end{proof}
}
For instance, it follows from Theorem~\ref{thm:agg-hard} that,
whenever $q$ is a non-hierarchical CQ without self-joins, it is
$\fpsharpp$-complete to compute the Shapley value for the
aggregate-relational queries $\qcnt\bracs{q}$,
$\qsum\angs{\varphi}\bracs{q}$, $\qmax\angs{\varphi}\bracs{q}$, and
$\qmin\angs{\varphi}\bracs{q}$, unless $\varphi(\tup c)=0$ for all databases
$D$ and tuples $\tup c\in q(D)$. Additional examples follow.

\begin{exa}
  Consider the numerical query $\alpha_3$ from
  Example~\ref{example:agg}. Since $q_4$ is not hierarchical,
  Theorem~\ref{thm:agg-hard} implies that computing
  $\shapley(D,\alpha_4,f)$ is $\fpsharpp$-complete. Actually,
  computing $\shapley(D,\alpha,f)$ is $\fpsharpp$-complete for any
  non-constant aggregate-relational query over $q_4$. Hence, computing
  the Shapley value w.r.t.~$\qcnt\bracs{q_4}$ (which counts the
  papers in $\rel{Citations}$ with an author from California) or
  w.r.t.~$\qmax\angs{[2]}\bracs{q_4}$ (which calculates the number of
  citations for the most cited paper by a Californian) is
  $\fpsharpp$-complete as well.\qed
\end{exa}


Interestingly, it turns out that Theorem~\ref{thm:agg-hard} captures
precisely the hard cases for computing the Shapley value w.r.t.~any
summation over CQs without self-joins. In particular, the following
argument shows that $\shapley(D,\qsum\angs{\varphi}\bracs{q},f)$ can be
computed in polynomial time if $q$ is a hierarchical CQ without self-joins. Let $q=q(\tup x)$ be an arbitrary CQ. For $\tup a\in q(D)$, let
$q_{[\tup x\rightarrow \tup a]}$ be the Boolean CQ obtained from $q$
by substituting every head variable $x_j$ with the value of $x_j$ in
$\tup a$. Hence, we have that $\qsum\angs{\varphi}\bracs{q}(D)=\sum_{\tup
  a\in q(D)} \left[\varphi(\tup a)\cdot q_{[\tup x\rightarrow \tup a]}(D)\right]$. The
linearity of the Shapley value (stated as a fundamental property in
Section~\ref{sec:shapley}) implies that:
\begin{equation}\label{eq:additive-sum}
\shapley(D,\qsum\angs{\varphi}\bracs{q},f)
=
\sum_{\tup a\in q(D)}
\varphi(\tup a)\cdot
\shapley(D,q_{[\tup x\rightarrow \tup a]},f)\,.
\end{equation}
Then, from Theorem~\ref{thm:bcq-dichotomy} we conclude that if $q$ is
a hierarchical CQ with self-joins, then  $\shapley(D,q_{[\tup
  x\rightarrow \tup a]},f)$ can be computed in polynomial time for every $\tup a\in q(D)$. Hence,
we have the following corollary of Theorem~\ref{thm:bcq-dichotomy}.

\def\sumdichotomy{ Let $q$ be a hierarchical CQ without self-joins.
  If $\alpha$ is an aggregate-relational query
  $\qsum\angs{\varphi}\bracs{q}$, then $\shapley(D,\alpha,f)$ can be
  computed in polynomial time, given $D$ and $f$ as input. In particular,
  $\shapley(D,\qcnt\bracs{q},f)$ can be computed in polynomial time.
}

\begin{cor}\label{cor:sum-dichotomy}
\sumdichotomy
\end{cor}
Together with Theorem~\ref{thm:agg-hard}, we get a full dichotomy for
$\qsum\angs{\varphi}\bracs{q}$ over CQs without self-joins.

The complexity of computing $\shapley(D,\alpha,f)$ for other
aggregate-relational queries remains an open problem for the general
case where $q$ is a hierarchical CQ without self-joins. We can,
however, state a positive result for $\qmax\angs{\varphi}\bracs{q}$ and
$\qmin\angs{\varphi}\bracs{q}$ for the special case where $q$ consists of
a single atom (i.e., aggregation over a single relation).

\def\maxsingleatom{ Let $q$ be a CQ with a single atom. Then,
   $\shapley(D,\qmax\angs{\varphi}\bracs{q},f)$ and
  $\shapley(D,\qmin\angs{\varphi}\bracs{q},f)$ can be computed in polynomial
  time.  }
\begin{prop}\label{prop:max}
\maxsingleatom
\end{prop}
{\color{black}
\begin{proof}
Since $q$ consists of a single atom, each fact adds at most one answer to the query result in each permutation. Let $\tup a_f\in q(D)$ be the answer corresponding to the fact $f$ (i.e., $q(\set{f})=\set{\tup a_f}$). Let $\sigma$ be a permutation. First, if there exists an exogenous fact $f'$ such that $\varphi(q(\set{f'}))>\varphi(\{\tup a_f\})$, then $f$ will never affect the maximum value, and $\shapley(D,\qmax\angs{\varphi}[q],f)=0$. Hence, from now on we assume that this is not the case. If $\tup a_f$ already appears in the query result before adding the fact $f$ (that is, $\tup a_f\in q(\sigma_f)$), then clearly $f$ does not affect the maximum value. If $\tup a_f$ is added to the query result only after adding $f$ in the permutation (that is, $q(\sigma_f\cup\set{f}\cup D\exo)\setminus q(\sigma_f\cup D\exo)=\set{\tup a_f}$), then $f$ only affects the maximum value if $\varphi(\{\tup a_f\})>\max_{\tup a\in q(\sigma_f\cup D\exo)}\varphi(\{\tup a\})$. In this case, it holds that 
  \[v(\sigma_f\cup\set{f}\cup D\exo)-v(\sigma_f\cup D\exo)=\varphi(\{\tup
  a_f\})-\max_{\tup a\in q(\sigma_f\cup D\exo)} \varphi(\{\tup a\})\]

Let $V=\set{v_1,\dots,v_m}$ be the set of values associated with the answers in $q(D)$ (that is, $V$ contains every value $v_j$ such that $\varphi(\{\tup a\})=v_j$ for some $\tup a\in q(D)$). Note that it may be the case that $\varphi(\{\tup a_1\})=\varphi(\{\tup a_2\})$ for $\tup a_1\neq \tup a_2$; hence, it holds that $|V|\le|q(D)|$. For each value $v_j$ we denote by $n_{v_j}^<$ the number of endogenous facts $f'$ in the database that correspond to an answer $\tup a$ (i.e., $q(\set{f'})=\set{\tup a}$) such that $\varphi(\{\tup a\})<v_j$, and by $n_{v_j}^=$ the number of endogenous facts in the database that correspond to an answer $\tup a$ such that $\varphi(\{\tup a\})=v_j$. We also denote by $n_{v_j}^\le$ the number $n_{v_j}^<+n_{v_j}^=$.

Let $\set{v_{i_1},\dots,v_{i_k}}$ be the set of values in $V$ such that 
  \[\big(\max_{\tup a\in q(D\exo)}\varphi(\{\tup a\})\big)\le
  v_{i_r}<\varphi(\{\tup a_f\})\]
Let us assume, without loss of generality, that $\max_{\tup a\in q(D\exo)}\varphi(\{\tup a\})=v_{i_1}$.
For each $r\in\set{2,\dots,k}$ and for each $t\in\set{1,\dots,n_{v_{i_r}}^\le}$, we compute the number of
permutations $\sigma$ in which $\sigma_f$ contains $t$ facts, and it holds that $\max_{\tup a\in q(\sigma_f\cup D\exo)}\varphi(\{\tup a\})=v_{i_r}$:
  \[P_r^t=t!\cdot (N-t-1)!\sum_{\ell=1}^{\min{(n_{v_{i_r}}^=,t)}} {n_{v_{i_r}}^=
  \choose \ell}{n_{v_{i_r}}^<\choose t-\ell}\]
(That is, we choose at least one fact $f'$ such that $\varphi(q(\set{f'}))=v_{i_r}$ and then we choose the rest of the facts among the facts $f''$ such that $\varphi(q(\set{f''}))<v_{i_r}$).
We count the number of such permutations separately for $v_{i_1}$, because in this case, we do not have to choose at least one endogenous fact $f'$ such that $\varphi(q(\set{f'}))=v_{i_1}$ (as this is already the maximum value on the exogenous facts). Hence, the number of permutations in this case is:
  \[P_1^t=t!\cdot (N-t-1)!{n_{v_{i_1}}^{\le}\choose t}\]

The contribution of each such permutation to the Shapley value of $f$ is:
  \[v(S\cup\set{f})-v(S)=\varphi(\{\tup a_f\})-v_{i_r}\]
Thus, the total contribution of the permutations $\sigma$ such that $\max_{\tup a\in q(\sigma_f)}\varphi(\{\tup a\})=v_{i_r}$ to the Shapley value of $f$ is $(\varphi(\{\tup a_f\})-v_{i_r})\sum_{t=1}^{n_{v_{i_r}}^{\le}}P_r^t$.

Finally, the Shapley value of $f$ is:
  \[\shapley(D,\qmax\angs{\varphi}[q],f)=\frac{1}{|D\endo|!}\sum_{r=1}^{k}\left\{
    (\varphi(\{\tup a_f\})-v_{i_r})\sum_{t=1}^{n_{v_{i_r}}^{\le}}P_r^t\right\}\]

A similar computation works for $\qmin\angs{\varphi}[q]$. The main difference is that now we are looking to minimize the value; hence, instead of considering $\varphi(\{\tup a_f\})-\max_{\tup a\in q(\sigma_f\cup D\exo)} \varphi(\{\tup a\})$, we now use $\varphi(\{\tup a_f\})-\min_{\tup a\in q(\sigma_f\cup D\exo)} \varphi(\{\tup a\})$, and we only consider permutations where $\min_{\tup a\in q(\sigma_f\cup D\exo)} \varphi(\{\tup a\})>\varphi(\{\tup a_f\})$.
\end{proof}
}

As an example, if $\alpha$ is the query
$\qmax\angs{[2]}[q]$, where $q$ is given by $q(x,y)\dl\rel{Citations}(x,y)$, then we can compute in polynomial
time $\shapley(D,\alpha,f)$, determining the responsibility of each
publication (in our running example) to the maximum number of
citations.

\revtwo{The arguments in the proof of Proposition~\ref{prop:max} heavily rely on the assumption that each fact adds at most one answer to the query result; hence, we can refer to \e{the} answer associated with a certain fact. Moreover, this answer is independent of the permutation. However, this assumption does not hold for general queries, where the addition of a fact can add multiple answers, and the added set of answers depends on the other facts in the permutation.
Hence, the proof does not easily generalize to maximum and minimum over hierarchical queries consisting of more than one atom. We also cannot use here the linearity of expectation that was used to obtain a dichotomy for summation. Therefore, a complete classification of the complexity for general aggregate queries remain an open problem.}

\subsection{Approximation}\label{sec:approx}
In computational complexity theory, a conventional feasibility notion
of arbitrarily tight approximations is via the \e{Fully
  Polynomial-Time Approximation Scheme}, FPRAS for short.  Formally,
an FPRAS for a numeric function $f$ is a randomized algorithm
$A(x,\epsilon,\delta)$, where $x$ is an input for $f$ and
$\epsilon,\delta\in(0,1)$, that returns an $\epsilon$-approximation of
$f(x)$ with probability $1-\delta$ (where the probability is over the
randomness of $A$) in time polynomial in $x$, $1/\epsilon$ and
$\log(1/\delta)$. To be more precise, we distinguish between an
\e{additive} (or \e{absolute}) FPRAS:
\[\Pr{f(x)-\epsilon\leq A(x,\epsilon,\delta)\leq  f(x)+\epsilon)}\geq 1-\delta\,,\]
and a \e{multiplicative} (or \e{relative}) FPRAS:
\[\Pr{\frac{f(x)}{1+\epsilon}\leq A(x,\epsilon,\delta)\leq (1+\epsilon)f(x)}\geq 1-\delta\,.\]

Using the Chernoff-Hoeffding bound, we easily get an additive FPRAS of
$\shapley(D,q,f)$ when $q$ is \e{any} Boolean query
computable in polynomial time, by simply taking \revtwo{the average value}
over $O(\log(1/\delta)/\epsilon^2)$ trials of the following
experiment:
\begin{enumerate}
    \item Select a random permutation $(f_1,\dots,f_n)$ over the set of all endogenous facts.
    \item Suppose that $f=f_i$, and let
      $D_{i-1}=D\exo\cup\set{f_1,\dots,f_{i-1}}$. \revtwo{Return $q(D_{i-1}\cup\set{f})-q(D_{i-1})$.}
\end{enumerate}
In general, an additive FPRAS of a function $f$ is not necessarily a
multiplicative one, since $f(x)$ can be very small. For example, we
can get an additive FPRAS of the satisfaction of a propositional
formula over Boolean i.i.d.~variables by, again, sampling the
averaging, but there is no multiplicative FPRAS for such formulas
unless $\mbox{BPP}=\mbox{NP}$. Nevertheless, the situation is
different for $\shapley(D,q,f)$ when $q$ is a CQ, since the Shapley
value is never too small (assuming data complexity).
\begin{prop}\label{prop:large-shap}
  Let $q$ be a fixed Boolean CQ. There is a polynomial $p$ such that for
  all databases $D$ and endogenous facts $f$ of $D$ it is the case
  that $\shapley(D,q,f)$ is either zero or at least $1/(p(|D|))$.
\end{prop}
\begin{proof}
  We denote $m=|D\endo|$. If there is no subset $S$ of $D\endo$ such
  that $f$ is a counterfactual cause for $q$ w.r.t.~$S$, then
  $\shapley(D,q,f)=0$. Otherwise, let $S$ be a minimal such set (i.e.,
  for every $S'\subset S$, we have that $(S'\cup D\exo)\not\models q$). Clearly, it holds that $S\le k$, where $k$ is the number of atoms of $q$. The probability to choose a permutation $\sigma$, such
  that $\sigma_f$ is exactly
  $S\setminus\set{f}$ is $\frac{(|S|-1)!(m-|S|)!}{m!}\ge\frac{(m-k)!}{m!}$ (recall that $\sigma_f$ is the set of facts that appear before $f$ in $\sigma$). Hence, we have
  that $\shapley(D,q,f)\ge \frac{1}{(m-k+1)\cdot ...\cdot m}$,
  and that concludes our proof.
\end{proof}
It follows that whenever $\shapley(D,q,f)=0$, the above additive approximation is also zero, and when $\shapley(D,q,f)>0$, the additive FPRAS also provides a multiplicative FPRAS. Hence, we have the following.
\begin{cor}\label{cor:fpras}
For every fixed Boolean CQ, the Shapley value has both an additive and a multiplicative FPRAS.
\end{cor}

\paragraph*{Approximation for Aggregate Queries}
\revtwo{The Chernoff-Hoeffding bound applies to the additive approximation of 
any function with a ``bounded domain,'' that is, where the gap between the maximal and minimal value is polynomial in the size of the input. Hence, we immediately conclude that there is an additive FPRAS for $\qcnt$. We also get an additive FPRAS for sum, average, median (or any quantile), min and max in the case where the values are from a bounded domain.} 

What about multiplicative approximation for aggregate queries? 
Interestingly, Corollary~\ref{cor:fpras} generalizes to a
multiplicative FPRAS for summation (including counting) over CQs.  By
combining Corollary~\ref{cor:fpras} with
Equation~\eqref{eq:additive-sum}, we immediately obtain a
multiplicative FPRAS for $\shapley(D,\qsum\angs{\varphi}[q],f)$, in the case where all the
features $\varphi(\tup a)$ in the summation have the same sign (i.e.,
they are either all negative or all non-negative). In particular,
there is a multiplicative FPRAS for $\shapley(D,\qcnt[q],f)$.

\begin{cor}\label{cor:fpras-sum}
  For every fixed CQ $q$, $\shapley(D,\qsum\angs{\varphi}[q],f)$ has a
  multiplicative FPRAS if either $\varphi(\tup a)\geq 0$ for all $\tup
  a\in q(D)$ or $\varphi(\tup a)\leq 0$ for all $\tup a\in q(D)$.
\end{cor}

Observe that the above FPRAS results allow the CQ $q$ to have self-joins. This is in contrast to the complexity results we established in
the earlier parts of this section, regarding exact evaluation. In
fact, an easy observation is that Proposition~\ref{prop:large-shap}
continues to hold when considering \e{unions of conjunctive queries}
(UCQs). Therefore, Corollaries~\ref{cor:fpras} and~\ref{cor:fpras-sum}
remain correct in the case where $q$ is a UCQ.

The existence of additive and multiplicative approximations for other aggregate queries remains an open problem.

\def\causale#1{\mathrm{CE}(#1)}
\def\Li{\mathrm{\Phi}}

\section{Related Measures}\label{sec:measures}
\revtwo{
In this section, we discuss our work in comparison to some alternative measures for the responsibility of tuples to database queries. 
}

\subsection*{Causal responsibility}
Causality and causal responsibility~\cite{Pearl:2009:CMR:1642718,Halpern:2016:AC:3003155} have been applied in data management, defining a fact as a {\em cause}  for a query result as follows: For an instance $D=D\exo\cup D\endo$, a fact ${f \in D\endo}$ is an {\em actual cause} for a Boolean CQ $q$,
if there exists ${\Gamma \subseteq D\endo}$, called a {\em contingency set} for $f$,  such that $f$ is a  counterfactual cause for $q$ in ${D\smallsetminus \Gamma}$~\cite{DBLP:journals/pvldb/MeliouGMS11}. 
 The { responsibility} of an actual cause $f$ for $q$ is defined by ${\rho(f) \ := \ \frac{1}{|\Gamma| + 1}}$, where
${|\Gamma|}$ is the
size of a smallest contingency set for $f$. If $f$ is not an actual cause, then $\rho(f)$ is zero~\cite{DBLP:journals/pvldb/MeliouGMS11}.
 Intuitively, facts with higher responsibility provide stronger explanations.\footnote{These notions can be applied to any monotonic query (i.e., whose answer set can only grow when the database grows, e.g., UCQs and Datalog queries)~\cite{DBLP:journals/mst/BertossiS17,DBLP:journals/ijar/BertossiS17}.}

\begin{exa}\label{ex:cause} Consider the database of our running example, and the query $q_1$ from Example~\ref{example:cq}. The fact $\auf_1$ is an actual cause with minimal contingency set $\Gamma =\{\auf_2,\auf_3,\auf_4\}$. So, its responsibility is $\frac{1}{4}$. Similarly, $\auf_2$, $\auf_3$ and $\auf_4$ are actual causes with responsibility $\frac{1}{4}$.
\end{exa}

\begin{exa} \label{ex:datalogq2} Consider the database $G$ and the query $p_{ab}$ from Example~\ref{ex:datalogq}. All facts in $G$ are actual causes since every fact appears in a path from $a$ to $b$.
It is easy to verify that all the facts in $D$ have the same causal responsibility, {$\frac{1}{3}$}, which may be considered as counter-intuitive given that $e_1$ provides a direct path from $a$ to $b$.
\end{exa}

As shown in Example~\ref{ex:datalogq}, the Shapley value gives a more
intuitive degree of contribution of facts to the query result than
causal responsibility.  Actually, Example~\ref{ex:datalogq} was used
in~\cite{DBLP:conf/tapp/SalimiBSB16} as a motivation to introduce an
alternative to the notion of causal responsibility, that of \e{causal
  effect}.
  
\subsection*{Causal effect}
To quantify the contribution of a fact to the query result, Salimi et
al.~\cite{DBLP:conf/tapp/SalimiBSB16} view the database as a
tuple-independent probabilistic database where the probability of each
endogenous fact is $0.5$ and the probability of each exogenous fact is
$1$ (i.e., it is certain).  The \e{causal effect} of a fact $f \in
D\endo$ on a numerical query $\alpha$ (in particular, a Boolean query) is a difference of expected values
\cite{DBLP:conf/tapp/SalimiBSB16}:
$$\causale{D,\alpha,f} \eqdef\,\, \mathbb{E}(\alpha(D)\mid f) - 
\mathbb{E}(\alpha(D)\mid \neg f)\,.
$$
where $f$ is the event that the fact $f$ is present in the database,
and $\neg f$ is the event that the fact $f$ is absent from the
database.

\begin{exa}\label{ex:causalEff2} Consider again the database of
  our running example, and the query $q_1$ from
  Example~\ref{example:cq}. We compute $\causale{D,q_1,\auf_1}$. It
  holds that: $\mathbb{E}(q_1\mid \neg\auf_1) = 0 \cdot P(q_1 =0\mid
  \neg\auf_1) + 1 \cdot P(q_1 =1\mid \neg\auf_1) =
  1-P(\neg\auf_2\wedge\neg\auf_3\wedge\neg\auf_4) = \frac{7}{8}$.
  Similarly, we have that $\mathbb{E}(q_1\mid \auf_1) = P(q_1 =1\mid
  \auf_1) = 1$.  Then, $\causale{D,q_1,\auf_1} = 1-\frac78
  =\frac{1}{8}$. Using similar computations we obtain that
  $\causale{D,q_1,\auf_2}=\causale{D,q_1,\auf_3}=\causale{D,q_1,\auf_4}=\frac18$.

  For $G$ and $p_{ab}$ of Example~\ref{ex:datalogq}, we have that
  $\causale{G,p_{ab},e_1} = 0.65625$, $\causale{G,p_{ab},e_2} =
  \causale{G,p_{ab},e_3} = 0.21875$, $\causale{G,p_{ab},e_4} =
  \causale{G,p_{ab},e_5} = \causale{G,p_{ab},e_6} = 0.09375$.  \qed
\end{exa}

Although the values in the two examples above are different from the Shapley values computed in Example~\ref{ex:bcq-ptime} and Example~\ref{ex:datalogq}, respectively, if we order the facts according to their contribution to the query result, we will obtain the same order in both cases. Note that unlike the Shapley value, for causal effect the sum of the values over all facts is not equal to the query result on the whole database. In the next example we consider aggregate queries.

\begin{exa} Consider the query $\alpha_1$ of
  Example~\ref{example:agg}. If $\auf_1$ is in the database, then the
  result can be either $20$, $28$, or $40$. If $\auf_1$ is absent,
  then the query result can be either $0$, $8$, or $20$. By computing
  the expected value in both cases, we obtain that
  $\causale{D,\alpha_1,\auf_1}=20$. Similarly, it holds that
  $\causale{D,\alpha_1,\auf_2}=\causale{D,\alpha_1,\auf_4}=1$,
  and $\causale{D,\alpha_1,\auf_3}=14$.\qed
\end{exa}


Interestingly, the causal effect coincides with a well known
wealth-distribution function in cooperative games, namely the
\e{Banzhaf Power Index}
(BPI)~\cite{Leech1990,10.2307/3689345,Kirsch2010}. This measure is
defined similarly to the definition of the Shapley value in
Equation~\eqref{eq:shapley-AB}, except that we replace the ratio
$\frac{|B|!\cdot (|A|-|B|-1)!}{|A|!}$ with $\frac{1}{2^{|A|-1}}$.

\def\propbanzhaf{ Let $\alpha$ be a numerical query, $D$ be a database, and
  $f\in D\endo$. Then,
$$\causale{D,\alpha,f} = \frac{1}{2^{|D\endo|-1}} \cdot \sum_{E \subseteq (D\endo\setminus \{f\})} \left[\alpha(D\exo\cup E \cup \{f\}) - \alpha(D\exo\cup E)\right]$$
Hence, the causal effect coincides with the BPI.  }
\begin{prop} \label{prop:banzhaf} 
\propbanzhaf
\end{prop}
{\color{black}
\begin{proof}
The following holds.
\begin{align*}
    \causale{D,\alpha,f}&=E(\alpha(D)\mid f)-E(\alpha(D)\mid\neg f)\\
    &\stackrel{(*)}{=}\sum_{E \subseteq (D\endo\setminus \{f\})} \frac{1}{2^{|D\endo|-1}} \cdot\alpha(D\exo\cup E \cup \{f\}) - \sum_{E \subseteq (D\endo\smallsetminus \{f\})} \frac{1}{2^{|D\endo|-1}} \cdot\alpha(D\exo\cup E)\\
    &=\frac{1}{2^{|D\endo|-1}} \cdot \sum_{E \subseteq (D\endo\setminus \{f\})} \alpha(D\exo\cup E \cup \{f\}) - \alpha(D\exo\cup E)
\end{align*}
The transition $(*)$ is correct since every endogenous fact in the
probabilistic database has probability $0.5$ and they are all independent; hence, all the possible
worlds have the same probability $\frac{1}{2^{|D\endo|-1}}$. (Recall
that we condition on $f$ being either present or absent from the
database, and all exogenous facts are certain; thus, the probability
of each possible world depends only on the facts in
$D\endo\setminus\set{f}$.) Then, for each $E\subseteq
D\endo\setminus\set{f}$, it holds that $\alpha(D\exo\cup E \cup
\{f\})$ is the value of the query on the possible world that contains
all the exogenous facts, the fact $f$, and all the endogenous facts in
$E$, but does not contain the endogenous facts in $D\endo\setminus
(E\cup\set{f})$. Hence, $\sum_{E \subseteq (D\endo\setminus \{f\})}
\frac{1}{2^{|D\endo|-1}} \cdot\alpha(D\exo\cup E \cup \{f\})$ is by
definition the expected value $E(\alpha(D)\mid f)$. Similarly,
$\sum_{E \subseteq (D\endo\smallsetminus \{f\})}
\frac{1}{2^{|D\endo|-1}} \cdot\alpha(D\exo\cup E)$ is the expected
value $E(\alpha(D)\mid\neg f)$, and that concludes our proof.
\end{proof}
}

\revtwo{In the next section we will discuss in more detail the complexity of the causal effect.}

\revtwo{
\subsection*{SHAP score}
One of the instantiations of the Shapley value is the \e{SHAP score} that has been used in the context of machine learning
for explaining the prediction of a model~\cite{NIPS2017_8a20a862}. This score could be applied to the attribution of responsibility to tuples in query answering, and it would give a measure that is similar in spirit, yet technically different, from our application of the Shapley value. Both apply the Shapley value in a cooperative game where the players are the endogenous facts. The difference is in the definition of the cooperative game. In our case, the players of a coalition, as a subinstance of the database, occur in the database, while the others are excluded from the computation of the wealth function, which is the answer of the query on the resulting subinstance. In the case of the SHAP score as applied to our setting, we could view the database as a tuple-independent probabilistic database where the facts of the coalition are deterministic (probability one) and the others are probabilistic (say with the probability $\frac12$), and the wealth function is the \e{expectation} of the query answer on the resulting probabilistic database. 

The complexity of the SHAP score (in a more abstract setting than  query answering) has been studied by Van den Broeck et
al.~\cite{van2021tractability} and by Arenas et
al.~\cite{arenas2021tractability,arenas2021complexity}.\footnote{The referenced work has been published later than the conference version of this article~\cite{DBLP:conf/icdt/LivshitsBKS20}.}
Immediate algorithmic approaches that one can derive to compute the SHAP score for Boolean CQs and UCQs are \e{(a)} to 
reduce the problem to probabilistic query answering~\cite{van2021tractability} \e{and (b)} to compile the \e{provenance} of the query into a \e{deterministic and decomposable
circuit} and apply 
 to the circuit
the Shapley computation of  Arenas et
al.~\cite{arenas2021tractability,arenas2021complexity}. Yet, albeit the similarity between our Shapley value and the SHAP score, they are different enough that we do not see an immediate way of directly translating results (e.g., via simple reductions) from one to the other. It is sensible, though, that we could apply similar techniques, namely reduction to/from probabilistic query answering and knowledge compilation, to derive our results and perhaps more general results. We leave this investigation to future research.
}

{\color{black}
\subsection{The Complexity of the Causal-Effect Measure (and Banzhaf Power Index)}\label{sec:banzhaf}
We now show that the complexity results obtained in this work for the exact computation of the Shapley value also apply to the causal effect (and BPI). These results are, in fact, easier to obtain, via a connection to probabilistic databases~\cite{DBLP:series/synthesis/2011Suciu}. The extension of the approximation results will be discussed later.

\begin{prop}
Theorem~\ref{thm:bcq-dichotomy}, Theorem~\ref{thm:agg-hard}, Corollary~\ref{cor:sum-dichotomy}, and Proposition~\ref{prop:max} hold for the Banzhaf Power Index.
\end{prop}
\begin{proof}
We start with Theorem~\ref{thm:bcq-dichotomy}---the dichotomy for Boolean self-join-free CQs. The positive side of the dichotomy (polynomial-time computation for hierarchical queries) is obtained via a reduction to query evaluation over probabilistic databases. Recall that Dalvi and Suciu~\cite{10.1145/2395116.2395119} showed that this problem is solvable in polynomial time for hierarchical Boolean CQs without self-joins. For a Boolean query, we have that $\causale{D,\alpha,f}$ is the probability that $q$ is true when $f$ is present in the database minus the probability that $q$ is true when $f$ is absent from the database. The first probability can be computed via a reduction to query evaluation over probabilistic databases by defining the probability of $f$ to be $1$ (and leaving the probabilities of the rest of the facts intact), and the second probability can be computed by a similar reduction after removing $f$ from the database altogether.

As for the negative side of the dichotomy (hardness of computation for non-hierarchical queries), we again use the result of Dalvi and Suciu~\cite{10.1145/2395116.2395119}. This time, we construct a reduction from their problem to our problem for the query $\cqrst$. They showed that query evaluation over probabilistic databases is $\fpsharpp$-complete for $\cqrst$, even when the probabilities of all the facts in the database are either $1$ or $0.5$. Hence, given an input database $D$ to their problem, we construct an input database $D'$ to our problem by adding all the facts of $D$, as well as two facts $R(\val{a})$ and $T(\val{b})$ with probability $1$ and a fact $S(\val{a},\val{b})$ with probability $0.5$, where $\val{a}$ and $\val{b}$ are two fresh constants that do not appear in $D$. Clearly, if $S(\val{a},\val{b})$ is present in the database, then the probability of $q$ being true is $1$. On the other hand, if $S(\val{a},\val{b})$ is absent, then the probability of $q$ being true is exactly the probability of $q$ being true on the original $D$. Therefore, to obtain this probability, we simply need to compute the value $1-\causale{D',q,S(\val{a},\val{b})}$. To prove hardness for any other non-hierarchical query we again use the result of Lemma~\ref{lemma:reduction-from-cqrst} that clearly also applies to the BPI. An important observation here is that in the reduction of Lemma~\ref{lemma:reduction-from-cqrst} for the BPI we preserve the probabilities of the facts ($0.5$ for endogenous facts and $1$ for exogenous facts).

Since the Shapley value and the BPI entail the same complexity for Boolean CQs, it is rather straightforward that
Theorem~\ref{thm:agg-hard} and Corollary~\ref{cor:sum-dichotomy} hold for the BPI as well, and the proof is almost identical (with the main difference being the fact that for the BPI we consider subsets of facts rather than permutations); hence, we do not give the proofs again here. Note that in the proof of the tractability side of Corollary~\ref{cor:sum-dichotomy} we do not use any special property of the Shapley value, but rather the linearity of expectation, and so this proof applies to the BPI as well. Finally, the proof of Proposition~\ref{prop:max} is also very similar for the BPI, where, again, instead of counting permutation, we count subsets of facts.
\end{proof}

Note that the dichotomy of Dalvi and Suciu~\cite{10.1145/2395116.2395119} for query evaluation over probabilistic databases also applies to queries with self-joins. \revtwo{The tractable cases of their dichotomy immediately translate to the causal effect measure, as in the case of self-join-free queries.} However, the probabilities used in the hardness proofs for such queries are not necessarily $0.5$ and $1$; hence, unlike the case of self-join-free queries, the hardness results of Dalvi and Suciu~\cite{10.1145/2395116.2395119} for queries with self-joins cannot be directly translated to complexity results for the causal-effect measure that, as aforementioned, assumes a probability of $0.5$ for endogenous facts and a probability of $1$ for exogenous facts.

\revtwo{In a recent work~\cite{DBLP:conf/pods/KenigS21}, Kenig and Suciu showed that the hard cases of the dichotomy for queries with self-joins remain hard even when all the probabilities are $0.5$ or $1$. However, even this stronger result does not immediately translate to a dichotomy for the causal effect measure. In particular, for self-join-free queries, we could provide an alternative proof for non-hierarchical queries, that does not go through the query $q_\mathsf{RST}$; however, this proof is slightly more involved. To use the same idea that we use in the proof of hardness for $q_\mathsf{RST}$ for other non-hierarchical queries $q$, we first need to identify a connected component $q'$ of $q$ that is non-hierarchical, and then construct the reduction from probabilistic query answering for the query $q'$. This method cannot be used in general for queries with self-joins, as a homomorphism from the $q$ to $D$ might map several atoms of $q$ to the same fact of $D$. We could, however, use this technique to obtain hardness for some queries with self-joins. In particular, we can prove hardness for connected queries that do not use constants. A similar idea has been used in~\cite{DBLP:conf/pods/ReshefKL20} in the proof of Theorem 5.1. A complete classification of the complexity for queries with self-join remains an open problem.}

Finally, we discuss the approximate computation of the causal-effect measure. Similarly to case of the Shapley value, we can easily obtain an additive approximation via Monte Carlo Sampling. Using this method we can obtain an additive approximation of $E(\alpha(D)\mid f)$, as well as an additive approximation of $E(\alpha(D)\mid\neg f)$. Clearly, by subtracting the second value from the first one we obtain an additive approximation of $E(\alpha(D)\mid f)-E(\alpha(D)\mid\neg f)$ (with approximation factor $2\epsilon$). The story is different for multiplicative approximation, as we cannot obtain a multiplicative approximation of $x-y$ from such approximation of $x$ and $y$, unless the property of Proposition~\ref{prop:large-shap} holds. The following example shows that this is not the case for the causal-effect measure.

\begin{exa}
Consider the query $\cqrst$ and a database consisting of $n+1$ triplets of facts: $R(a_i),S(a_i,b_i),T(b_i)$ for $i\in\set{1,\dots,n+1}$, where all the facts in $R$ and $T$ are exogenous and the facts of $S$ are endogenous. Let $f=S(a_1,b_1)$. Clearly, we have that $q(D\exo\cup E\cup\set{f})-q(D\exo\cup E)=1$ only if none of the other endogenous facts appears in $E$ (i.e., $E=\emptyset$). Therefore, we have that $\causale{D,q,f}=\frac{1}{2^n}$.\qed
\end{exa}

Note that the above example does not preclude the existence of a multiplicative approximation for the causal-effect measure (and BPI). This is left open for future investigation. 
}

\section{Conclusions}\label{sec:conclusions}
We introduced the problem of quantifying the contribution of database
facts to query results via the Shapley value. We investigated the
complexity of the problem for Boolean CQs and for aggregates over
CQs. Our dichotomy in the complexity of the problem establishes that
computing the exact Shapley value is often intractable. Nevertheless,
we also showed that the picture is far more optimistic when allowing
approximation with strong precision guarantees.
{\color{black} Finally, we showed that all of our complexity results for the Shapley value, except for the existence of a multplicative approximation, also hold for the causal-effect measure (and Banzhaf power index).}

Many questions, some quite fundamental, remain for future
investigation. While we have a thorough understanding of the complexity
for Boolean CQs without self-joins, very little is known in the
presence of self-joins. For instance, the complexity is open even for
the simple query $q()\dl R(x,y),R(y,z)$.  We also have just a partial
understanding of the complexity for aggregate functions over CQs,
beyond the general hardness result for non-hierarchical queries
(Theorem~\ref{thm:agg-hard}). In particular, it is important to
complete the complexity analysis for maximum and minimum, and to
investigate other common aggregate functions such as average, median,
percentile, and standard deviation. 
{\color{black}An additional question that is still open is whether there exists a multiplicative approximation for the causal-effect measure.}
Another direction is to
investigate whether and how properties of the database, such as low
treewidth, can reduce the (asymptotic and empirical) running time of
computing the Shapley value. Interestingly, the implication of a low
treewidth to Shapley computation has been studied for a different
problem~\cite{DBLP:conf/ijcai/GrecoLS15}.

\section*{Acknowledgments}
  \noindent The work of Ester Livshits, Benny Kimelfeld and Moshe Sebag was supported by the Israel Science Foundation (ISF), grants 1295/15 and 768/19, and the Deutsche Forschungsgemeinschaft (DFG) project 412400621 (DIP program). The work of Ester
Livshits was also supported by the Technion Hiroshi Fujiwara Cyber
Security Research Center and the Israel Cyber Bureau. Leopoldo Bertossi is a member of RelationalAI's Academic Network, and has been partially funded by  the ANID - Millennium Science Initiative Program - Code ICN17-002.


\bibliography{main}
\bibliographystyle{alpha}

\end{document}